\RequirePackage{fix-cm}
\documentclass[smallextended,natbib]{svjour3}       
\smartqed  
\usepackage{graphicx}
\usepackage{fullpage,url}
\usepackage{amsmath,amsfonts,amssymb}
\usepackage{color} 
\usepackage{adjustbox}
\usepackage{multirow}
\usepackage{threeparttable} 
\usepackage{makecell} 
\usepackage{multirow}
\usepackage{textcomp}
\usepackage{algorithm, algorithmicx}
\usepackage{algpseudocode}
\usepackage{rotating,epsfig}

\newcommand{\bx} {{\bf x}}
\newcommand{\by} {{\bf y}}
\newcommand{\bX} {{\bf X}}

\newcommand{\bs}{\boldsymbol}

\newcommand{\DP}{\displaystyle}

\newcommand{\argmin}{\operatornamewithlimits{argmin}}
\newcommand{\eps}{\epsilon}

\newcommand{\BE}{\begin{equation}}
\newcommand{\EE}{\end{equation}}
\newcommand{\BEN}{\begin{equation}\nonumber}
\newcommand{\BA}{\begin{array}}
\newcommand{\EA}{\end{array}}
\newcommand{\BPA}{\begin{pmatrix}\begin{array}}
\newcommand{\EAP}{\end{array}\end{pmatrix}}
\newcommand{\BI}{\begin{itemize}}
\newcommand{\EI}{\end{itemize}}

\newcommand{\RR}{\mathbb{R}}
\newcommand{\hsigma}{\hat{\sigma}}
\newcommand{\homega}{\hat{\omega}}

\newcommand{\hbeta}{\hat{\beta}}
\newcommand{\hbetabs}{\hat{\bs{\beta}}}

 \newlength{\myindent}
\begin{document}

\title{An efficient GPU-Parallel Coordinate Descent Algorithm 
for Sparse Precision Matrix Estimation via Scaled Lasso
}

\titlerunning{GPU-Parallel CD for SPMESL}        

\author{Seunghwan Lee        \and
        Sang Cheol Kim \and
        Donghyeon Yu*      
}


\institute{
  							Seunghwan Lee \and {Donghyeon Yu} 
  							\at 
  							Department of Statistics, Inha University, Inchoen, Korea\\
  						  \email{dyu@inha.ac.kr}           
           \and
           Sang Cheol Kim \at
              Division of Bio-Medical Informatics, Center for Genome Science, National Institute of Health, Korea Centers for Disease Control
and Prevention, Cheongju, Korea
}

\date{Received: date / Accepted: date}

\maketitle

\renewcommand{\baselinestretch}{1.2}.
\begin{abstract}
	\noindent The sparse precision matrix plays an essential role in the Gaussian graphical model since a zero off-diagonal element 
indicates conditional independence of the corresponding two variables given others. In the Gaussian graphical model, many methods have been  
proposed and their theoretical properties are given as well. Among  
these, the sparse precision matrix estimation via scaled lasso (SPMESL) has an attractive feature to which the penalty level is automatically set to achieve the optimal convergence rate under the sparsity and invertibility conditions. 
Conversely,  
other methods need to be used in searching for the optimal tuning parameter. Despite  
such an advantage, the SPMESL has not been widely used due to its expensive computational cost. In this paper,
we develop a GPU-parallel coordinate descent (CD) algorithm for the SPMESL and numerically show that the the proposed algorithm is much faster than the least angle regression (LARS) tailored to the SPMESL.
Several comprehensive numerical studies are conducted to investigate the scalability of the proposed algorithm and the estimation performance of the SPMESL. The results show that the SPMESL has the lowest false discovery rate for all cases and the best performance in the case where the level of the sparsity of the columns is high.
\end{abstract}

\section{Introduction}

The covariance matrix and its inverse are of main interest in multivariate analysis to model dependencies between variables.
Traditionally, these two parameters 
have been estimated by
the sample covariance matrix and its inverse based on
the maximum likelihood (ML) estimation.
Although these ML estimators are often asymptotically unbiased and simple to calculate,
there are some weaknesses when the number of variables is greater than that  of samples. This circumstance is also known as high-dimensional low-sample size (HDLSS) data.
For HDLSS data, it is known that the sample covariance matrix
becomes inefficient \citep{Yao2015}.
In addition, the inverse of the sample covariance matrix is undefined since the sample covariance matrix is singular for HDLSS data.
For the covariance matrix, many methods have been proposed 
 under some structural
conditions such as bandable structure and sparse structure
to improve the estimation efficiency \citep{Bickel2008a, Bickel2008b, Rothman2009, Wu2009, Cai2010, Cai2011a, Cai2012}.

To obtain  an estimator of the \textit{precision matrix} (i.e., inverse covariance matrix) for 
HDLSS data,
various approaches have been developed by adopting the sparsity
assumption that there are many zero elements in the matrix.
The existing methods can be categorized  
into four approaches:
covariance estimation-induced approach, ML-based approach,
regression approach, and constrained $\ell_1$-minimization approach.
The covariance estimation-induced approach considers an indirect estimation using 
the inversion
of the well-conditioned shrinkage covariance matrix estimators
and applies multiple testing procedures to identify nonzero elements of the precision matrix (GeneNet, \cite{Schafer2005}).
The ML-based approach directly estimates the precision matrix
by maximizing the penalized likelihood function \citep{Yuan2007, Friedman2008, Witten2011, Mazumder2012}. 
The regression approach considers the linear regression model
and uses the fact that the nonzero regression coefficients correspond
to the nonzero off-diagonal elements of the precision matrix \citep{MB2006, Peng2009, Yuan2010, Sun2013, Ren2015, Khare2015, Ali2017}.
The constrained $\ell_1$-minimization (CLIME) approach \citep{Cai2011b} obtains the sparse precision matrix by solving the linear programming problem with the constraints
on the proximity to the precision matrix where the objective 
function is the sum of absolute values of the design variables.
The adaptive CLIME \citep{Cai2016} improves the CLIME and attains
the optimal rate of convergence.

Among the existing methods, 
sparse precision matrix estimation
via the scaled Lasso (SPMESL) proposed by \cite{Sun2013} is a
tuning-free procedure. 
Conversely, other
existing methods require searching the optimal tuning parameter, 
 GeneNet \citep{Schafer2005} and the neighborhood selection \citep{MB2006}
require choosing the level of the false discovery rate; 
other penalized methods
require choosing the optimal penalty level.
In addition,  SPMESL supports the consistency of the precision matrix
estimation under the sparsity and invertibility conditions that
are weaker than the irrepresentable condition \citep{VB2009, Sun2013}.
However, it has not been  
widely used for the sparse precision matrix
estimation due to the inefficiency of the implemented algorithm of the SPMESL using
the least angle regression (LARS) algorithm \citep{Efron2004}.
Even though the LARS algorithm efficiently provides a whole
solution path of the Lasso problem \citep{Tibshirani1996},
its computational cost is expensive and 
the SPMESL needs to solve $p$ independent Lasso problems
as the subproblems of the SPMESL.
Thus, the scalability of the SPMESL is still challenging.

In this paper, we found the possibility to improve the computational efficiency 
of the SPMESL with the empirical observation that the SPMESL does not need
a whole solution path of the Lasso problem based on the tuning-free
characteristic of the SPMESL (see details in Section 3). 
Motivated by this empirical observation, we propose a more efficient algorithm based
on the coordinate descent (CD) algorithm and the warm start strategy
for the scaled Lasso and the SPMESL
as applied to the standard lasso problem in \cite{Wu2008}.
Moreover, we develop the \emph{row-wise updating parallel CD
algorithm using graphics processing units} (GPUs) adequate for
the SPMESL.  
To efficiently implement the proposed parallel CD algorithm,
  we consider \emph{the active response matrix} that consists
  of columns of active response variables that corresponds to
the error variance estimate not converged at the current iteration.
We will show the efficiency of the proposed algorithms using
comprehensive numerical studies.
In this paper, we also provide the numerical comparisons of the estimation performance of
the SPMESL with three different penalty levels, which are theoretically suggested in previous literature \citep{Sun2012, Sun2013} but there is no comparison result
of them in terms of the estimation performance.

The remainder of the paper is organized as follows.
Section 2 introduces the SPMESL and its original algorithm.
We explain the proposed CD algorithm
and GPU-parallel CD algorithm in Section 3.
Section 4 provides a comprehensive numerical study, including the comparisons of computation times and estimation performances
with other existing methods.
We provide a summary of the paper and discussion
in Section 5.

\section{Sparse Precision Matrix Estimation via Scaled Lasso}

In this section, we briefly introduce the scaled Lasso,
the SPMESL, and their algorithms.

\subsection{Scaled Lasso}

The scaled Lasso proposed by \cite{Sun2012} is a variant of a penalized regression model with the Lasso penalty.
To be specific, let $\by \in \RR^n$, $\bX \in \RR^{n\times p}$,
$\bs{\beta} \in \RR^p$, 
and $\sigma >0$. Consider a linear regression model
$\by = \bX \bs{\beta} + \epsilon$, where $\epsilon \sim N(0, \sigma^2 {\bf I}_n)$ and ${\bf I}_n$  is the $n$-dimensional identity matrix.
The scaled Lasso considers the minimization of the following objective function $L_{\lambda_0}(\bs{\beta}, \sigma)$: 
\begin{equation} \label{eqn:sc}
L_{\lambda_0} (\bs{\beta},\sigma)
= \frac{\|\by - \bX \bs{\beta}\|_2^2}{2n\sigma} + \frac{\sigma}{2}  + \lambda_0 \|\bs{\beta}\|_1,
\end{equation} 
where $\lambda_0$ is a given tuning parameter and $\|\bx\|_1 = \sum_{j=1}^p |x_j|$ for $\bx \in \RR^p$.
Thus, the scaled Lasso simultaneously obtains the estimate
$\hat{\bs{\beta}}$ of the regression coefficient $\bs{\beta}$ and
the estimate $\hat{\sigma}$ of the error standard deviation $\sigma$.
It is proved that the objective function $L_{\lambda_0}$
is jointly convex in $(\bs{\beta},\sigma)$ and is strictly
convex for $\sigma$ in \cite{Sun2012}.
From the convexity of the objective function,
the solution $(\hat{\bs{\beta}}, \hat{\sigma})$ of
the scaled Lasso problem can be obtained by
the block CD algorithm as follows:

\begin{itemize}

\item[] (Step 1) With a fixed $\hsigma$, consider the minimization of 
$\hsigma L_{\lambda_0}(\bs{\beta}, \hsigma)$: 
\begin{equation} \nonumber
\hsigma L_{\lambda_0}(\bs{\beta}, \hsigma) = \frac{\|\by -\bX \bs{\beta}\|_2^2}{2n}
+ \hsigma\lambda_0 \|\bs{\beta}\|_1 + \frac{\hsigma^2}{2}.
\end{equation}
The minimizer $\hbetabs$ of $\hsigma L_{\lambda_0}(\bs{\beta},\hsigma)$ can be obtained by solving
the following standard lasso problem with $\lambda=\hsigma \lambda_0$:
\begin{equation} \label{eq:lasso}
\min_{\bs{\beta}} \frac{\|\by -\bX \bs{\beta}\|_2^2}{2n}
+ \lambda \|\bs{\beta}\|_1.
\end{equation}

\item[] (Step 2) With a fixed $\hbetabs$, the minimizer
$\hsigma$ of $L_{\lambda_0}(\hbetabs,\sigma)$ is easily obtained by
\begin{equation} \nonumber
\hsigma = \frac{\|\by - \bX \hbetabs\|_2}{\sqrt{n}}.
\end{equation}

\item[] (Step 3) Repeat Steps 1 and 2 until convergence occurs.

\end{itemize}

\noindent In the original paper of the scaled Lasso,
the standard lasso problem is solved by the LARS algorithm
\citep{Efron2004}, which provides a whole solution path
of the standard Lasso problem. 
During the block CD algorithm, 
 the minimizer
$\hbetabs(\hsigma^{(r)})$
of $\hsigma^{(r)} L_{\lambda_0}(\bs{\beta},\hsigma^{(r)})$
in Step 1 at the $r$-th iteration is
obtained from the solution path of the
standard Lasso problem with $\lambda=\hsigma^{(r)}\lambda_0$.

In addition to the joint estimation of $\bs{\beta}$ and $\sigma$,
the scaled Lasso has attractive features as described in \cite{Sun2012}.
First, the scaled Lasso guarantees the consistency
of $\hbetabs$ and $\hsigma$ under two conditions: the penalty level condition
$\lambda_0 > A \sqrt{2n^{-1}\log p}$
for $A > 1$ and the compatibility condition,
which implies the oracle inequalities for the prediction and
estimation \citep{VB2009}.
Second, the scaled Lasso estimates are scale-equivariant in
$\by$ in the sense that $\hbetabs(\bX,\alpha \by) = \alpha \hbetabs(\bX,\by)$ and $\hsigma(\bX, \alpha \by) = |\alpha| \hsigma(\bX,\by)$.
Finally, the authors suggest using the universal penalty
level $\lambda_{0}^{U} = \sqrt{2n^{-1}\log p}$ for $\lambda_0$
based on their numerical and real-data examples.
We can consider the scaled Lasso with
the universal penalty level as a tuning-free procedure.
Note that the universal penalty level does not satisfy
the theoretical requirement $\lambda_0 > A\sqrt{2n^{-1}\log p}$
for the consistency of $\hsigma$.
The authors of the scaled Lasso provide
several conditions that can weaken the required condition
for $\lambda_0$, in order to
justify using a penalty level smaller than $A\sqrt{2n^{-1}\log p}$
for $A>1$ \citep{Sun2012}.

\subsection{Sparse precision matrix estimation via scaled Lasso}

Let $\bs{\Sigma}=(\sigma_{jk})_{1\le j,k\le p}$ and $\bs{\Omega}=\bs{\Sigma}^{-1}=(\omega_{jk})_{1\le j,k\le p}$  be a covariance
matrix and a corresponding precision matrix, respectively.
Suppose that $\bx^{(i)} = (X_{i1},\ldots, X_{ip})$  for
$i=1,2,\ldots,n$ are independently drawn from the multivariate normal distribution with mean ${\bf 0}$ and covariance matrix $\bs{\Sigma}$. 
Let $\bX \in \RR^{n\times p}$ be a data matrix and
 $\bx_k = (X_{1k},\ldots,X_{nk})^T$ be the $k$th column of $\bX$.
Consider following linear regression models for $k=1,\ldots,p$:
\begin{equation} \label{eq:reg}
X_{ik} = \sum_{l\neq k} \beta_{lk} X_{il} + \epsilon_{ik},~
\end{equation}
where $\epsilon_{ik}$s are independent and
identically distributed random variables drawn from
the normal distribution with mean 0 and variance $\sigma_k^2$.
As applied in the regression approach for the sparse
precision matrix estimation, the elements of the precision matrix
can be represented into the regression coefficient
and the error variance by using the following relationships:
\begin{equation} \label{eq:relation}
\omega_{jk} = -\frac{\beta_{jk}}{\sigma_k^2},
~ \omega_{kk} = \frac{1}{\sigma_k^2}
~\mbox{ for } 1\le j\neq k \le p.
\end{equation}
As the scaled Lasso
simultaneously estimates the regression coefficients $(\beta_{jk})$
and the error standard deviation $\sigma_k$ as described in Section 2.1, we can use the scaled Lasso to estimate the precision matrix.
Thus, the SPMESL considers solving the scaled Lasso problems column-wise
and defines the SPMESL estimator that combines the estimates
from the $p$ scaled Lasso problems.

To be specific, let  $\textbf{B} = (b_{jk})_{1\le j,k \le p}$  be
a matrix of the regression coefficients such that
$b_{jk} = \beta_{jk}$ for $j\neq k$ and
$b_{jj} = -1$ for $j=1,\ldots,p$.
Denote ${\bf b}^{(j)} = (b_{j1},
\ldots, b_{jp})$
and ${\bf b}_{k}=(b_{1k},\ldots,b_{pk})^T$ as the $j$th row and the $k$th column
of a matrix ${\bf B}$, respectively. 
We further let ${\bf S} = (s_{jk})$ be the sample covariance matrix.
Subsequently, the precision matrix can be represented with ${\bf B}$ and ${\bf D}={\rm diag}(\sigma_1^{-2}, \ldots, \sigma_p^{-2})$  as follows: 
\[
\bs{\Omega} = -{\bf B}{\bf D} = (-\sigma_1^{-2}{\bf b}_{1},\ldots,-\sigma_p^{-2}{\bf b}_{p}).
\] 
To obtain the estimate of  $\bs{\Omega}$, the SPMESL solves the following $p$ independent scaled Lasso problems first: for $k=1,\ldots,p$,
\begin{equation} \label{eq:spmesl}
(\hat{\bf b}_{k}, \hsigma_k) =
\argmin_{\bs{\beta}_k\in \RR^p:\beta_{kk}=-1, \sigma_k > 0}
~ \frac{\|\bX_k - \sum_{j\neq k} \beta_{jk} \bX_j\|_2^2}{2n\sigma_k}
+ \frac{\sigma_k}{2} + \lambda_0 \sum_{j\neq k} |\beta_{jk}|.
\end{equation}
Note that the SPMESL in \cite{Sun2013} originally considers 
$\lambda_0 \sum_{j\neq k} \sqrt{s_{jj}}|\beta_{jk}|$ 
instead of $\lambda_0 \sum_{j\neq k} |\beta_{jk}|$
to penalize the coefficients on the same scale.
In this paper, we assume that the columns of the data matrix $\bX$
are centered and scaled to $\bX_k^T \bX_k = n$ for $k=1,\ldots,p$.
Thus, $s_{jj} = 1$ for $j=1,\ldots,p$. 
This assumption does not affect the estimation performance of the precision matrix as the scaled Lasso has the scale-equivariant property in the response variable as explained
in the previous section. We can easily recover
the estimate $\hat{\bs{\Omega}}^o = (\hat{\omega}_{jk}^o)_{1\le j,k\le p}$ from the data in the original scale
by the following Proposition \ref{scale}:
\begin{proposition} \label{scale}
Let $\bX \in \RR^{n \times p}$ and
$\tilde{\bX} = \bX {\bf C}$ be a
data matrix and the scaled data matrix with 
${\bf C}={\rm diag}(s_{11}^{-1/2},\ldots,s_{pp}^{-1/2})$,  where  $s_{jj} > 0$ for $j=1,\ldots,p$. 
Denote $\hat{\bs{\Omega}}^o$ as
the estimate of the precision matrix by applying
the scaled Lasso in \eqref{eq:spmesl} with the penalty
term $\lambda_0 \sum_{j\neq k} \sqrt{s_{jj}}|\beta_{jk}|$ column by
column with $\bX$. Similarly, denote $\hat{\bs{\Omega}}^C = (\hat{\omega}_{jk}^C)_{1\le j,k \le p}$ as
the estimate by the scaled Lasso in \eqref{eq:spmesl}
with $\tilde{\bX}$.
Then,  $\hat{\bs{\Omega}}^o = {\bf C}\hat{\bs{\Omega}}^C {\bf C}$. 
\end{proposition}

\begin{proof}
By the definition of $\tilde{\bX}$, the $k$th column of $\tilde{\bX}$
is  $\tilde{\bX}_{k} = \bX_{k}/\sqrt{s_{kk}}$. 
Let $\hat{\bf B}^C = (\hat{b}_{jk}^C)_{1\le j,k \le p}$ and $(\hsigma_{k,C})_{1\le k \le p}$ be the solutions of
the $p$ scaled Lasso problem in \eqref{eq:spmesl} with $\tilde{\bX}$.
We further let 
$\hat{\bf B}^o = (\hat{b}_{jk}^o)_{1\le j,k \le p}$ and  $(\hsigma_{k,o})_{1\le k \le p}$ be the solutions of the following scaled Lasso problems: 
for $k=1,\ldots,p$, 
\begin{equation} \label{eq:org_spm}
(\hat{\bf b}_{k}^o, \hsigma_{k,o}) =
\argmin_{\bs{\beta}_k\in \RR^p:\beta_{kk}=-1, \sigma_k > 0}
~ \frac{\|\bX_k - \sum_{j\neq k} \beta_{jk} \bX_j\|_2^2}{2n\sigma_k}
+ \frac{\sigma_k}{2} + \lambda_0 \sum_{j\neq k} \sqrt{s_{jj}}|\beta_{jk}|.
\end{equation} 
By substituting  $\tilde{\bX}$ with $\bX {\bf C}$   in \eqref{eq:spmesl}
and the reparameterization with $\tilde{\beta}_{jk} = \beta_{jk}/\sqrt{s_{jj}}$ for $j\neq k$,
the $k$th scaled Lasso problem becomes 
\begin{equation} \label{eq:spmesl_sub} 
(\hat{\bf b}_{k}^C, \hsigma_{k,C}) =
\argmin_{\tilde{\bs{\beta}}_k\in \RR^p:\tilde{\beta}_{kk}=-1, \sigma_k > 0}
~ \frac{\|\bX_k/\sqrt{s_{kk}} - \sum_{j\neq k} \tilde{\beta}_{jk}\bX_j\|_2^2}{2n\sigma_k}
+ \frac{\sigma_k}{2} + \lambda_0 \sum_{j\neq k} \sqrt{s_{jj}}|\tilde{\beta}_{jk}|.
\end{equation}  
From the forms of the problems \eqref{eq:org_spm} and
\eqref{eq:spmesl_sub},
the relationship  $\hat{\bf b}^o_{k} = \sqrt{s_{kk}}\hat{\bf b}^C_{k}$ and $\hsigma_{k,o} = \sqrt{s_{kk}}\hsigma_{k,C}$  
hold by the scale equivariant property of the scaled Lasso.
In addition, by the reparameterization, 
$\hat{b}^o_{jk} = \sqrt{s_{kk}/s_{jj}} \hat{b}^C_{jk}$
for $1\le j, k \le p$. 
Combining the above relationships,
the $(j,k)$th element $\hat{\omega}_{jk}^o$ of the precision matrix estimate $\hat{\bs{\Omega}}^o$ is represented as 
\[
\hat{\omega}_{jk}^o = -\hat{b}_{jk}^o \hsigma_{k,o}^{-2}
= -(s_{jj}s_{kk})^{-1/2} \hat{b}_{jk}^{C} \hsigma_{k,C}^{-2}
= (s_{jj}s_{kk})^{-1/2} \hat{\omega}_{jk}^C.
\]  
Hence, $\hat{\bs{\Omega}}^o = {\bf C} \hat{\bs{\Omega}}^C {\bf C}$.
\end{proof}

Note that the result  $\hat{\bs{\Omega}}^o = {\bf C} \hat{\bs{\Omega}}^C {\bf C}$  in
Proposition \ref{scale} is
consistent with the property of  $({\rm Var}({\bf A} {\bf z}))^{-1} = {\bf A}^{-T} {\rm Var}({\bf z})^{-1} {\bf A}^{-1}$
for a $p$-dimensional random vector ${\bf z}$ and a nonsingular matrix
${\bf A}\in \RR^{p \times p}$.

After solving $p$ independent scaled Lasso problems,
the estimate  $\hat{\bs{\Omega}}_1 = - \hat{\bf B} \hat{\bf D} = (\hat{\omega}_{jk,1})_{1\le j,k \le p}$  of the precision matrix is
obtained. However, the estimate  $\hat{\bs{\Omega}}_1$  is not
symmetric in general. To find the symmetric estimate
of the precision matrix using the current estimate  $\hat{\bs{\Omega}}_1$, 
the SPMESL considers solving the following linear programming problem
as in \cite{Yuan2010}: 
\begin{equation}
\hat{\bs{\Omega}} = \argmin_{{\bf M}:{\bf M}^T = {\bf M}} \|{\bf M}- \hat{\bs{\Omega}}_1\|_1.
\end{equation} 
Remark that the authors of the SPMESL consider the above linear
programming problem for the symmetrization step in \cite{Sun2013},
but they applied the following symmetrization step in the
implemented R package \texttt{scalreg}:
\begin{equation} \label{eq:symm}
\hat{\omega}_{jk}=\hat{\omega}_{kj}
= \hat{\omega}_{jk,1}
I(|\hat{\omega}_{jk,1}|\leq|\hat{\omega}_{kj,1}|) +
\hat{\omega}_{kj,1}
I(|\hat{\omega}_{jk,1}| > |\hat{\omega}_{kj,1}|),
\end{equation}
which is applied in the CLIME 
and the theoretical properties are developed on this symmetrization
\citep{Cai2011b}.
In addition, for the high-dimensional data, the symmetrization applied
in the CLIME is favorable as
its computational cost is cheap and it is easily
parallelizable. For these reasons,
we apply
the symmetrization step \eqref{eq:symm} in the proposed algorithm.

As stated in the previous section, the scaled Lasso
guarantees the consistency of the regression coefficients
and the error variance under the compatibility condition.
Thus, the SPMESL also guarantees column-wise consistency
of  $\hat{\bs{\Omega}}_1$  under the compatibility conditions,
which is independently defined in each column of  $\hat{\bs{\Omega}}_1$, 
as the SPMESL applies the scaled Lasso column wise.
To derive the overall consistency of  $\hat{\bs{\Omega}}$, 
the authors of the SPMESL considers
the capped $\ell_1$ sparsity and the invertibility conditions as follows.
\begin{enumerate}
\item[(i)] Capped $\ell_1$ sparsity condition:
For a certain $\epsilon_0$, $\lambda_0^*$
not depending on $j$ and an index set $T_j \subset \{1,2,\ldots,p\} \setminus \{j\}$, the capped $\ell_1$ sparsity of the $j$th column with $t_{j}>0$ is
defined as
\[
|T_j| + \sum_{k\neq j, k \notin S_j} \frac{|\omega_{kj}|\sqrt{\sigma_{kk}}}{(1-\eps_0) \sqrt{\omega_{jj}}\lambda_{0}^*} \le a_j.
\]

\item[(ii)] Invertibility condition:
Let ${\bf W}={\rm diag}(\sigma_{11},\ldots, \sigma_{pp})$ and ${\bf R} = {\bf W}^{-1/2}\bs{\Sigma} {\bf W}^{-1/2}$.
Further, let $T_j \subseteq Q_j \subseteq \{1,2,\ldots,p\}\setminus \{j\}$. The invertibility condition
is defined as 
\[
\inf_j \left\{ \frac{ {\bf u}^T {\bf R}_{-j,-j} {\bf u}}{\| {\bf u}_{Q_j}\|_2^2} ~:~ {\bf u}\in \RR^p, {\bf u}_{Q_j} \neq 0, 1 \le j \le p\right\} \ge c_*
\] 
with a fixed constant $c_*>0$.
Note that the invertibility condition holds if the spectral norm of
${\bf R}^{-1}= {\bf D}^{1/2} \bs{\Omega} {\bf D}^{1/2}$ is bounded (i.e., $\|{\bf R}^{-1}\|_2 \le
c_*^{-1}$). 
\end{enumerate}

With some modifications on the capped $\ell_1$ sparsity and the invertibility conditions,
the authors of the SPMESL derive several conditions on $\lambda_{0}^*$ that
guarantees the estimation consistency of the precision matrix under the spectral norm.
Among them, for practical usage, we consider two conditions on a penalty level $\lambda_0 \ge \lambda_0^*$
as follows:
\begin{itemize}
\item Union bound for $p$ applications of the scaled Lasso (Theorem 2 in \cite{Sun2013}):
\begin{equation}
\lambda_0 = A \sqrt{4n^{-1} \log p} \mbox{ for } A > 1.
\end{equation}

\item Probabilistic error bound (Theorem 13 in \cite{Sun2013}):
\begin{equation}
\lambda_0 = A L_n (k/p) \mbox{ for } 1 < A \le \sqrt{2},
\end{equation}
where $k$ is a real solution of $k=L_1^4(k/p) + 2 L_1^2(k/p)$,
$L_n(t) = n^{-1/2} \Phi^{-1}(1-t)$, and $\Phi^{-1}(t)$ is the standard normal quantile function.
\end{itemize}

\noindent Note that a real solution of the equation $k-L_1^4(k/p) + 2 L_1^2(k/p)=0$ can easily be found
by applying the bisection method. For instance, we demonstrate two real solutions
for $p=100, 1000$ in Figure \ref{fig:pb} with the values of $k-L_1^4(k/p) + 2 L_1^2(k/p)$.
For $p=1000$ and $n=100$, $\lambda_{ub} = \sqrt{4n^{-1} \log p} \approx 0.5257$ and
$\lambda_{pb}=\sqrt{2} L_n (k/p) \approx 0.2810$ with $k=23.4748$ while
$\lambda_{univ} = \sqrt{2n^{-1} \log (p-1)} \approx 0.3717$,
where $\lambda_{ub}$ is the penalty level derived by the union bound,
$\lambda_{pb}$ is the penalty level derived by the probabilistic error bound,
and $\lambda_{univ}$ is the universal penalty level used in the scaled lasso.
In the paper of \cite{Sun2013}, the penalty level derived by the probabilistic error bound is suggested
for the SPMESL. However, there are no comparison results for the three penalty levels $\lambda_{univ}$, $\lambda_{ub}$, and $\lambda_{pb}$.
We conduct a comparison of performances for identifying the nonzero elements of $\bs{\Omega}$
by the three penalty levels above in Section 4 to provide a guideline for the penalty level $\lambda_0$.

\begin{figure}[!htb]
\centering
\begin{minipage}[r]{.45\linewidth}
\centerline{\epsfig{file=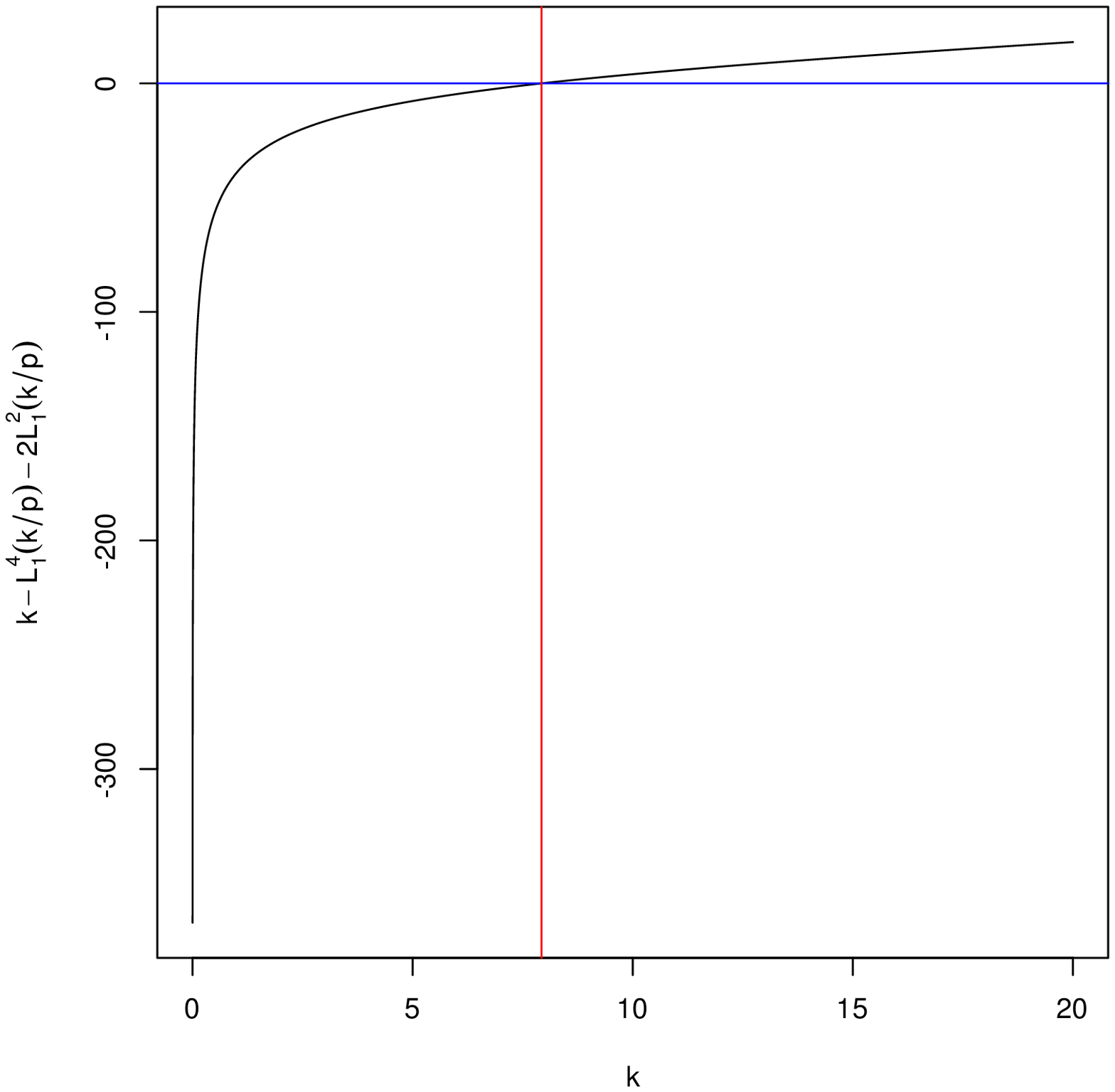,width=0.8\textwidth,height=0.24\textheight}}
\centerline{(1) $p=100$}
\end{minipage}
\begin{minipage}[l]{.45\linewidth}
\centerline{\epsfig{file=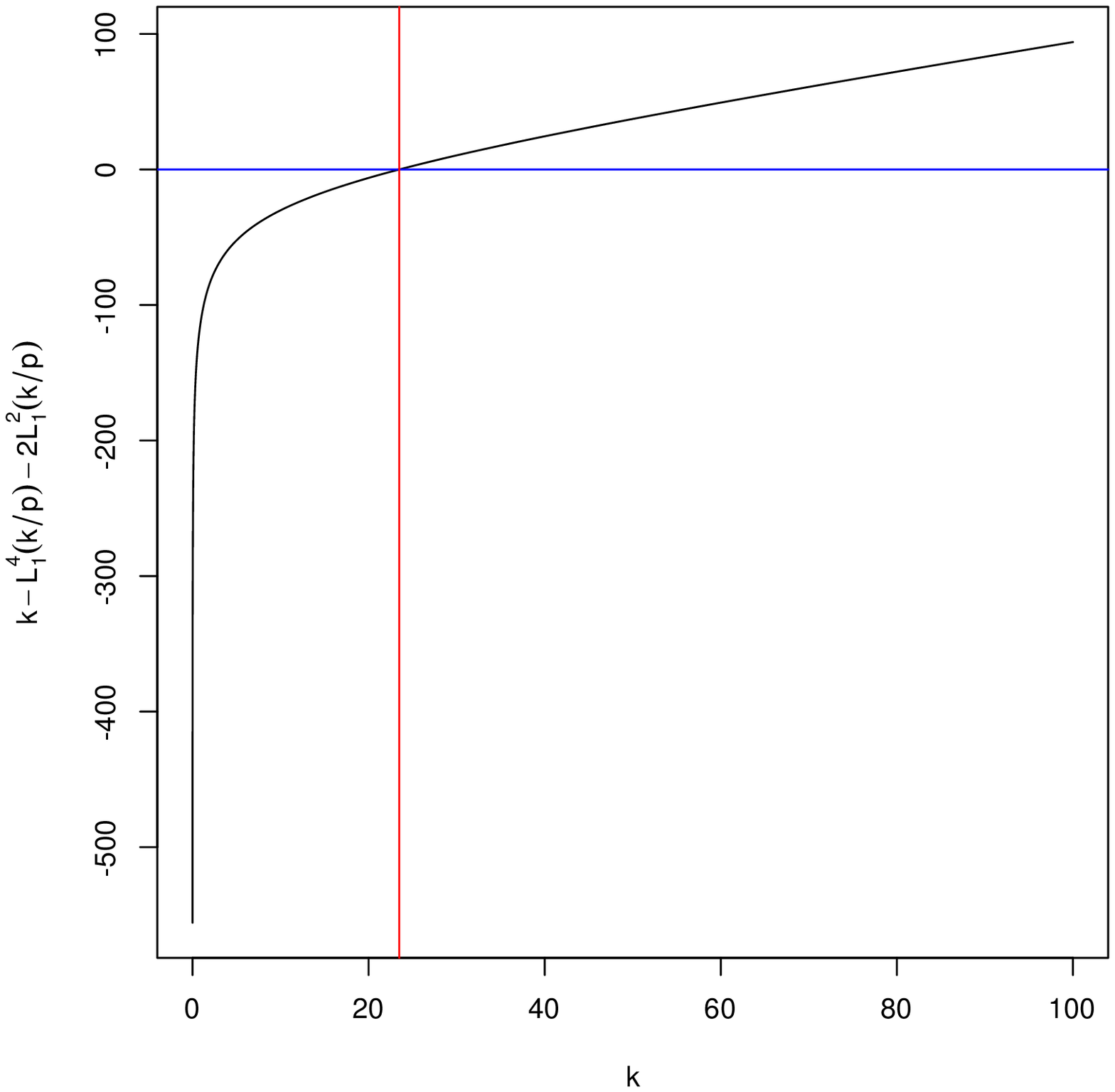,width=0.8\textwidth,height=0.24\textheight}}
\centerline{(2) $p=1000$}
\end{minipage}
\caption{Plots of $k-L_1^4(k/p) + 2 L_1^2(k/p)$ for $p=100, 1000$. Vertical red lines
denote the solutions of $k-L_1^4(k/p) + 2 L_1^2(k/p)=0$ obtained by the bisection method.} \label{fig:pb}
\end{figure}

\section{Efficient Coordinate Descent Algorithm for SPMESL and its GPU-parallelization}

The original algorithm for the scaled Lasso and the SPMESL
adopt the LARS algorithm, which provides
the  whole  solution path of the Lasso regression problem,
and its implemented R package \texttt{scalreg} is available in
the Comprehensive R Archive Network (CRAN) repository. 
As mentioned in the Introduction, 
we empirically observed that the block CD
algorithms for the scaled Lasso and the SPMESL do not need
a whole solution path of the standard Lasso problem in their sub-problems,
where the sub-problem denotes the minimization problem in Step 1 of the scaled Lasso problem. 
To describe our empirical observation, we consider an example with a linear regression model
$y_i = {\bf x}_i^T \bs{\beta} + \epsilon_i$ and $\epsilon_i \sim N(0,\sigma^2)$
for $i=1,2,\ldots,250$, where the true parameter
$\bs{\beta} = (\bs{\beta}_2^T, \bs{\beta}_{-1}^T, \bs{\beta}_0^T)^T \in \mathbb{R}^{500}$,
$\bs{\beta}_2 = (2,\ldots,2)^T \in \mathbb{R}^{5}$,
$\bs{\beta}_{-1} = (-1,\ldots,-1)^T \in \mathbb{R}^{5}$,
$\bs{\beta}_0=(0,\ldots,0)^T \in \mathbb{R}^{490}$,
and $\sigma = 3, 5$.
We set the initial value of $(\bs{\beta}, \sigma)$
as $({\bf 0}_{500\times 1}, 1)$.
As shown in Figure \ref{fig:conv_sig}, the numbers of iterations for the convergence of $\hat{\sigma}$ 
are less than 10
when the true parameter $\sigma = 3, 5$ and $\lambda_0=\sqrt{2n^{-1}\log p}$.
This implies that we do not need to obtain the whole
solution paths of $p$ lasso problems with respect to all $\lambda$ values.

\begin{figure}[!htb]
\centering
\begin{minipage}[r]{.45\linewidth}
\centerline{\epsfig{file=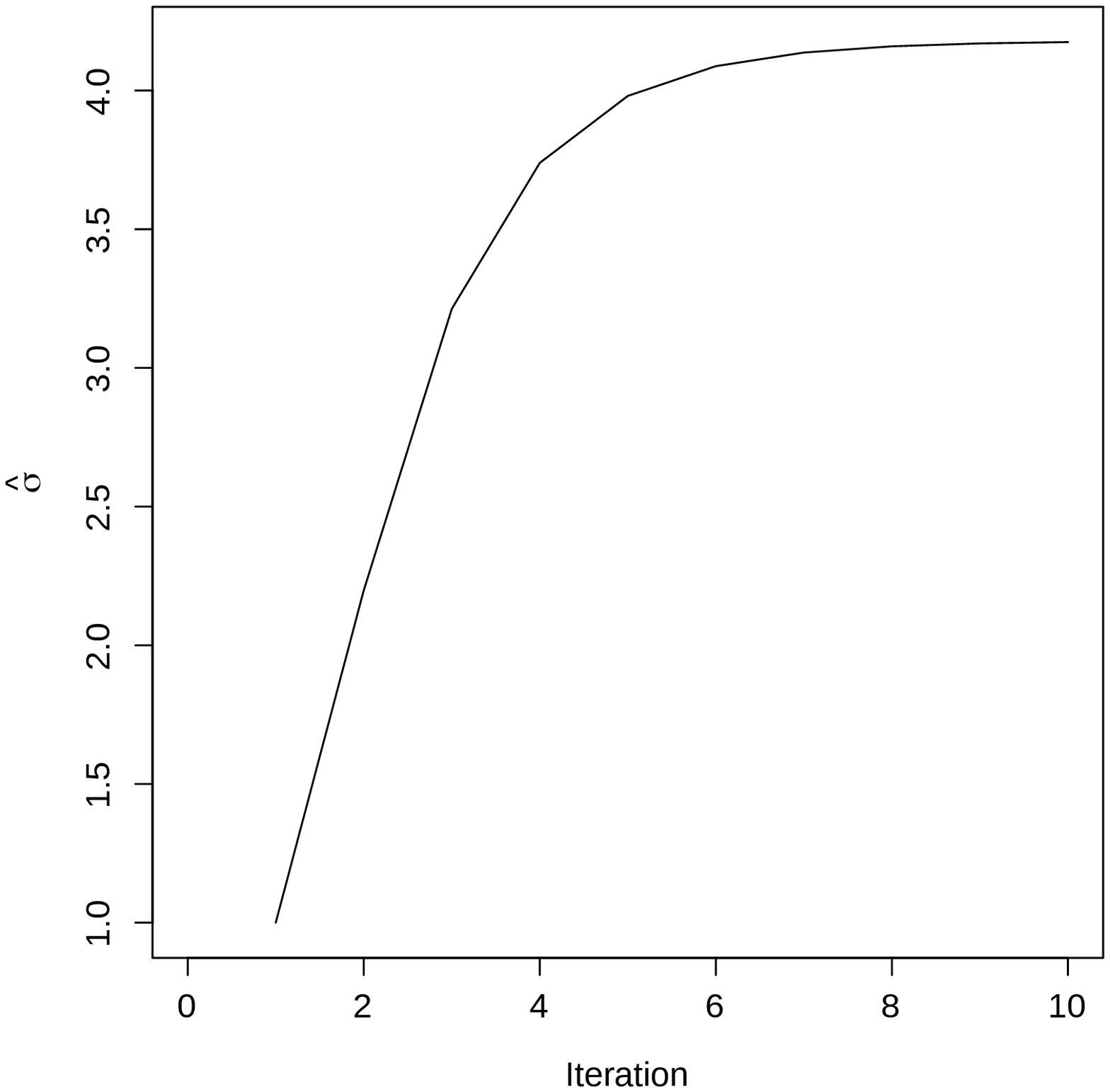,width=0.95\textwidth,height=0.24\textheight}}
\centerline{(1) $p=500, n=250, \sigma=3$}
\end{minipage}
\begin{minipage}[l]{.45\linewidth}
\centerline{\epsfig{file=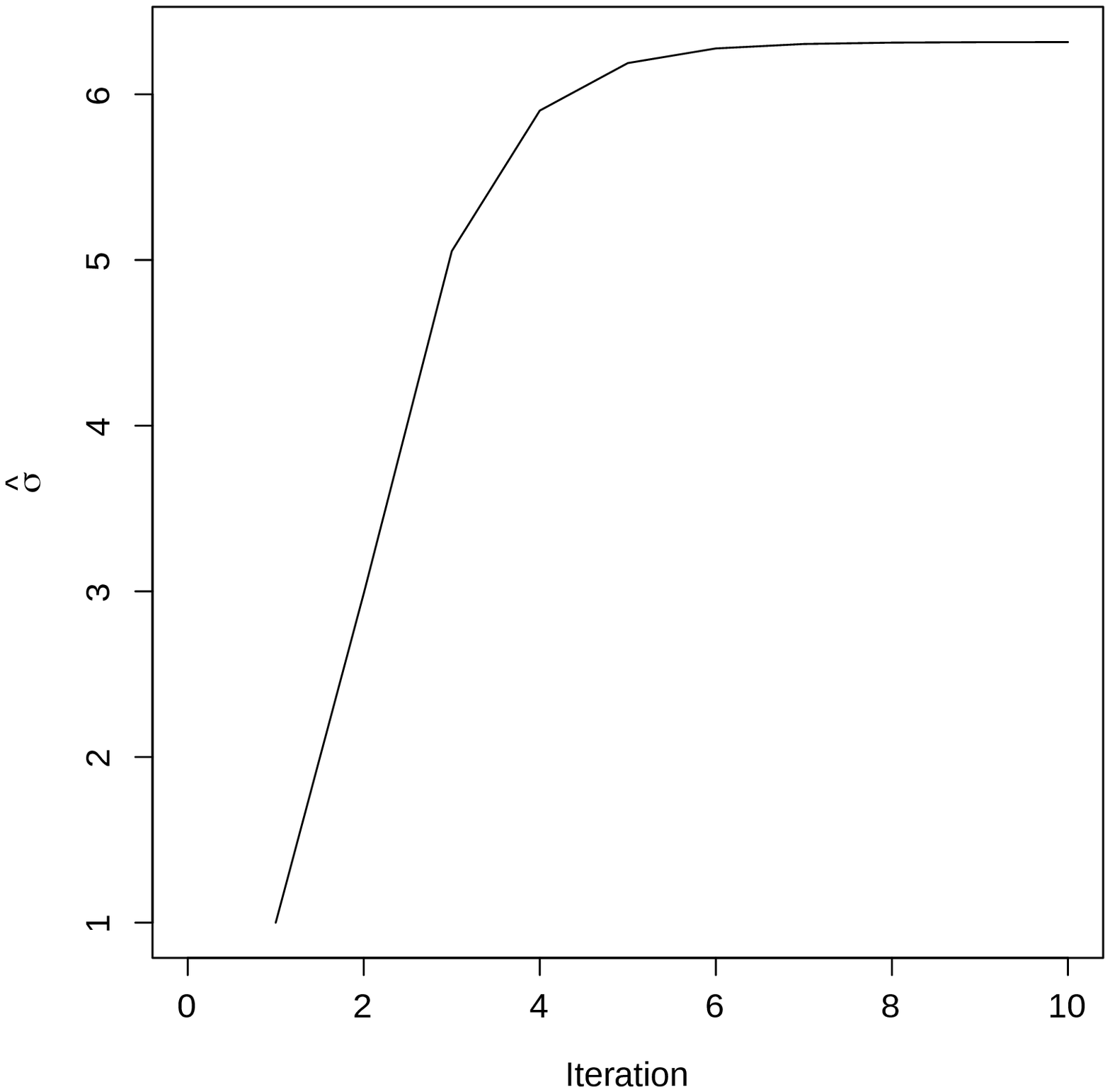,width=0.95\textwidth,height=0.24\textheight}}
\centerline{(2) $p=500, n=250, \sigma=5$}
\end{minipage}
\caption{Plots of the estimates of $\hat{\sigma}$ for $\sigma=3, 5$ along with the number of iterations.} \label{fig:conv_sig}
\end{figure}

Thus, the calculation of the whole solution path by the LARS algorithm
is inefficient for the scaled Lasso and the SPMESL.
In addition, the scaled Lasso and the SPMESL iteratively solve
the lasso problem with the penalty $\lambda_r = \hat{\sigma}^{(r)} \lambda_0$
in their inner iteration, where $\lambda_r$ denotes the penalty level at the $r$th iteration. As $\hat{\sigma}^{(r)}$ converges
to the minimizer of \eqref{eqn:sc}, the difference between
$\lambda_r$ and $\lambda_{r-1}$ decreases as the iteration proceeds.
This denotes that the Lasso estimate in the current iteration
is not far from that in the next iteration.
From this observation, the warm start strategy, which denotes
that the solution in the previous iteration is used as an initial value of the next iteration, is favorable and can efficiently accelerate
the algorithm for the Lasso regression problem in the inner iteration.

To fully utilize the warm start strategy, we consider the coordinate descent (CD) algorithm
for the Lasso regression problem in the inner iteration.
This is because it is known that the CD algorithm with the warm start strategy is
efficient for the Lasso
regression problem and has an advantage for memory consumption
\citep{Wu2008}.
In addition, the SPMESL needs solving $p$ independent scaled
Lasso problems to obtain the estimate of the precision matrix.
We develop the GPU-parallel CD algorithm for the SPMESL, which
updates $p$ coordinates simultaneously with GPUs.
In the following subsections, we introduce the CD algorithm for the
subproblem of the scaled Lasso and GPU-parallel CD algorithm for the SPMESL in detail.

\subsection{CD Algorithm for subproblem of the scaled Lasso}

In this subsection, we focus on the following subproblem of the scaled Lasso
with a given $\lambda_0$: 
\begin{equation} \label{eqn:lasso}
\hat{\bs{\beta}}^{(r)} = \argmin_{\bs{\beta}} \frac{1}{2n}\|\by - \bX \bs{\beta}\|_2^2 + \lambda^{(r-1)} \|\bs{\beta}\|_1,
\end{equation} 
where $\by \in \RR^n$ is a response vector, $\bX \in \RR^{n \times p}$
is a design matrix, $\lambda^{(r-1)} = \hat{\sigma}^{(r-1)} \lambda_0$,
and $\hat{\sigma}^{(r-1)} = \|\by - \bX\hat{\beta}^{(r-1)}\|_2/ \sqrt{n}$ is the iterative solution for $\sigma$ at the $(r-1)$th iteration. 
For
the notational simplicity, we use $\hat{\bs{\beta}}$ to denote the $r$th
iterative solution  $\hat{\bs{\beta}}^{(r)}$, and $\hat{\bs{\beta}}^{[cur]}
$ and $\hat{\bs{\beta}}^{[next]}$  denote the current
and the next iterative solution in the CD algorithm, respectively.
To apply the warm start strategy, we set the initial value  $\hat{\bs{\beta}}^{[cur]}$ for  $\hat{\bs{\beta}}$ to $\hat{\bs{\beta}}^{(r-1)}$. 
In this paper, we consider the cyclic CD algorithm
with an ascending order (i.e., coordinate-wise update from the smallest index to the largest index).
Subsequently, for $j=1,\ldots, p$, the CD algorithm updates $\hat{\beta}_j^{[next]}$
by following equations: 
\[
{\bf e}_j = \by - \sum_{l<j} \bX_l \hat{{\beta}}_l^{[next]}
- \sum_{l>j} \bX_l \hat{{\beta}}_l^{[cur]}, \quad
a_j = \bx_j^T {\bf e}_j / n + \hat{\beta}_j^{[cur]}, \quad
\hat{\beta}_j^{[next]} = \mbox{Soft}_\lambda(a_j),
\]
where $\mbox{Soft}_\lambda(a_j) = \mbox{sign}(a_j)(|a_j|-\lambda)_+$ is the soft-thresholding operator
and $(x)_+ = \max(x, 0)$.  
The CD algorithm repeats the cyclic updates until the convergence occurs,
where the common convergence criterion is the $L_\infty$-norm of
the difference between  $\hat{\bs{\beta}}^{[cur]}$ and $\hat{\bs{\beta}}^{[next]}$ 
(i.e., $\|\hat{\bs{\beta}}^{[next]} - \hat{\bs{\beta}}^{[cur]}\|_\infty$). 
The whole CD algorithm with warm start strategy for the scaled Lasso is summarized in Algorithm \ref{alg:cd}.

\begin{algorithm}[h]
\caption{CD algorithm with warm start strategy for the scaled Lasso}\label{alg:cd}
\begin{algorithmic}[1]
\Require ${\bf y}$, ${\bf X}$, $\lambda_0$, $\hat{\sigma}^{(0)}=1$,
$\hat{\bs{\beta}}^{(0)}={\bf 0}$,
convergence tolerance $\delta$. 
\Repeat{~~$r=0, 1, 2 \ldots$}  \smallskip 
\State $\lambda \gets \hat{\sigma}^{(r)}\lambda_0$  \Comment{Initialization of lasso subproblem}
 \smallskip 
\State $\hat{\bs{\beta}}^{[cur]} \gets \hat{\bs{\beta}}^{(r)}, \hat{\bs{\beta}}^{[next]} \gets \hat{\bs{\beta}}^{[cur]}$ 
\Comment{Warm start strategy}
 \smallskip 
\Repeat{~~$m = 0,1,2,\ldots$} 
 \smallskip
\State $\hat{\bs{\beta}}^{[cur]} \gets \hat{\bs{\beta}}^{[next]}$
 \smallskip
\For{$j = 1,\cdots, p$}
 \smallskip
\State ${\bf e}_j = \by - \sum_{l<j} \bX_l \hat{\bs{\beta}}^{[next]}
- \sum_{l>j} \bX_l \hat{\bs{\beta}}^{[cur]}$  \smallskip
\State $a_j = \bX_j^T {\bf e}_j / n + \hat{\beta}_j^{[cur]}$  \smallskip 
\State $\hat{\beta}_j^{[next]} = \mbox{Soft}_\lambda(a_j)$
 \smallskip
\EndFor
 \smallskip 
\Until{ $\|\hat{\bs{\beta}}^{[next]} - \hat{\bs{\beta}}^{[cur]}\|_\infty, < \delta$}  
\Comment{End of lasso subproblem}  \smallskip 
\State $\hat{\bs{\beta}}^{(r+1)} \gets \hat{\bs{\beta}}^{[next]}$ 
 \smallskip
\State $\displaystyle \hat{\sigma}^{(r+1)} =\frac{\|\by - \bX \hbetabs^{(r+1)}\|_2}{\sqrt{n}}$  \smallskip
\Until{ $|\hat{\sigma}^{(r+1)}-\hat{\sigma}^{(r)}| < \delta$} 
 \smallskip
\Ensure $\hat{\bs{\beta}} \gets \hat{\bs{\beta}}^{(r+1)}$, $\hat{\sigma} \gets \hat{\sigma}^{(r+1)}$ 
\end{algorithmic}
\end{algorithm}

\subsection{CD Algorithm for SPMESL}

As described in Section 2.2, the SPMESL estimates
the precision matrix by solving the $p$ scaled Lasso problems,
where each column of the observed data matrix is considered as the response variable and the other columns are considered as the exploratory variables. 
To be specific, let $(\bx^{(i)})^T = (X_{i1}, X_{i2}, \ldots, X_{ip})^T \sim N(0, \bs{\Omega}^{-1})$
be the $i$th random sample and $\bs{\Omega} = \bs{\Sigma}^{-1}$ be the precision matrix. 
Furthermore, let $\bx_k = (X_{1k},\ldots, X_{nk})^T$ be the $k$th column vector of the observed data matrix $\bX = (\bx_1, \ldots, \bx_p) \in \RR^{n\times p}$. 
The CD algorithm with the warm start strategy
for the SPMESL independently applies the CD algorithm in Section 3.1 to the subproblem \eqref{eq:spmesl} for $k=1,2,\ldots, p$.
The main procedures in the CD algorithm for the SPMESL are summarized
in the following two steps:
\begin{itemize}

\item {{\bf Updating $\hbetabs_{-k}$ for $1\le k \le p$:}}
Applying the CD algorithm with the warm-start strategy
for the following lasso subproblem:
for $k=1,2,\ldots,p$, 
\begin{equation} \label{eq:sublasso}
\hbetabs_{-k} = \argmin_{\bs{\beta}_{-k}~:~\beta_{kk}=0} \frac{\| {\bf x}_k - \bX ~\bs{\beta}_{-k} \|_2^2}{2n}
+ \hsigma_j \lambda_0 \|\bs{\beta}_{-k}\|_1,
\end{equation}
where $\bs{\beta}_{-k} = (\beta_{1,k},\ldots,\beta_{k-1,k},0,\beta_{k+1,k},\ldots,\beta_{p,k})^T \in \RR^{p}$.

\item {{\bf Updating $\hsigma_k$ for $1\le k \le p$}:}
For given $\lambda_0$ and $\hbetabs_{-k}s$, $\hsigma_k$s are obtained by the equation 
\begin{equation} \label{eq:sigma}
\hsigma_k =\frac{\| {\bf x}_k - \bX ~\hbetabs_{-k} \|_2}{\sqrt{n}},
\end{equation}
where $\hbetabs_{-k} = (\hbeta_{1,k},\ldots,\hbeta_{k-1,k},0,\hbeta_{k+1,k},\ldots,\hbeta_{p,k})^T$ 
is the solution to the problem \eqref{eq:sublasso}.

\end{itemize}
The CD algorithm for the SPMESL independently repeats the updating
$\hbetabs_{-k}$   and $\hsigma_k$ until convergence occurs for $k=1,2,\ldots, p$. 
After solving the $p$ scaled Lasso problems,
the final estimate  $\hat{\bs{\Omega}}$  of the precision matrix by the SPMESL is
obtained by the symmetrization \eqref{eq:symm}.
The whole CD algorithm with warm start strategy for the SPMESL is summarized in Algorithm \ref{alg:cd_spm}.

\begin{algorithm}[h]
\caption{CD algorithm with warm start strategy for the SPMESL}\label{alg:cd_spm}
\begin{algorithmic}[1]
\Require ${\bf X}$, $\lambda_0$, $\hat{\sigma}^{(0)}=1$, 
${\bf B} = (\hbeta_{ij}^{(0)})=(\hat{\bs{\beta}}_{-1}^{(0)}, \ldots, \hat{\bs{\beta}}_{-p}^{(0)})={\bf 0}$, 
convergence tolerance $\delta$.
\For{$k=1,\ldots, p$}  \smallskip 
\State ($\hbetabs_{-k}, \hsigma_{k}$) $\gets$ Apply Algorithm \ref{alg:cd} with ($\bx_k, \bX_{-k} = (\bx_1,\ldots, \bx_{k-1}, \bx_{k+1}, \ldots, \bx_p), \lambda_0$) 
 \smallskip
\EndFor
\For{$j=1,\ldots, p$} \Comment{Calculation of initial estimate of $\Omega$}
 \smallskip
\State $\homega_{jj} = \hsigma_j^{-2}$
\For{$k=1, \ldots, p$}  \smallskip
\If{$k \neq j$}  \smallskip
\State $\homega_{jk} = - \hbeta_{jk} {\homega}_{kk}$
 \smallskip
\EndIf
 \smallskip
\EndFor
\EndFor  \smallskip
\For{$j=1,\ldots, p-1$} \Comment{Symmetrization of $\Omega$}
 \smallskip
\For{$k=j+1, \ldots, p$}  \smallskip
\If{$|\homega_{jk}| > |\homega_{kj}|$}  \smallskip
\State $\homega_{jk} \gets \homega_{kj}$
\Else
\State $\homega_{kj} \gets \homega_{jk}$
 \smallskip
\EndIf
\EndFor
\EndFor
\Ensure $\hat{\bs{\Omega}} = (\homega_{jk})$
\end{algorithmic}
\end{algorithm}

\subsection{Parallel CD algorithm for SPMESL using GPU}

As we described in Section 3.2, the CD algorithm for the SPMESL solves $p$ independent scaled Lasso problems.
From this independence structure, $p$ elements in ${\bf B}=(\beta_{jk})$
can be updated in parallel, where $(p-1)$ elements
can simultaneously be updated in practice since $\beta_{kk}=0$ is fixed.
In addition, the computation of
$\hsigma$ is independent as well in the sense that
the update equation for $\hsigma_k$ only needs information
of  $\hbetabs_{-k}$. 
To describe the proposed parallel CD algorithm, we consider
the following joint minimization problem, which combines
$p$ scaled Lasso problems: 
\begin{equation} \label{eq:joint} 
(\widehat{\bf B}, \hat{\bs{\sigma}}) =
\argmin_{\{\bs{\beta}_{-k}, \sigma_k\}_{k=1}^p}
\sum_{k=1}^p \Big\{ \frac{\|\bx_k - \bX \bs{\beta}_{-k}\|_2^2}{2n\sigma_k}
+ \frac{\sigma_k}{2} + \lambda_0 \|\bs{\beta}_{-k}\|_1\Big\},
\end{equation}
where ${\bf B}=(\bs{\beta}_{-1},\ldots, \bs{\beta}_{-p})$ and $\bs{\beta}_{-k} = (\beta_{1,k},\ldots,\beta_{k-1,k},0,\beta_{k+1,k},\ldots,\beta_{p,k})^T$.  
As the updating equation \eqref{eq:sigma} for $\hsigma_k$ is simple
and easily parallelizable, we focus on the update for $\hat{\bf B}$
in this subsection. For the given iterative solution 
 $\hat{\bs{\sigma}}^{(r)}=(\hsigma_1^{(r)}, \ldots, \hsigma_p^{(r)})$, 
the subproblem for updating  $\hat{\bf B}^{(r+1)}$  can be
represented as follows: 
\begin{equation} \label{eq:joint_lasso}
\begin{array}{rcl}
\hat{\bf B}^{(r+1)}
&=& \DP \argmin_{\bs{\beta}_{-1},\ldots, \bs{\beta}_{-p}}
\sum_{k=1}^p g(\bs{\beta}_{-k};\hsigma_k^{(r)}, \lambda_0) =
\sum_{k=1}^p \Big\{ \frac{\|\bx_k - \bX \beta_{-k}\|_2^2}{2n} +
\lambda_k \|\bs{\beta}_{-k}\|_1\Big\}\\
&=& \DP
\argmin_{{\bf B}: b_{kk}=0,1\le k \le p} f({\bf B};\hat{\bs{\sigma}}^{(r)},\lambda_0) = \frac{1}{2n}\big\| \bX - \bX {\bf B}\big\|_F^2
+ \sum_{k=1}^p \lambda_k \|\bs{\beta}_{-k}\|_1,
\end{array}
\end{equation} 
where $\|{\bf A}\|_F = {\rm tr}({\bf A}^T {\bf A})=\big(\sum_{i,j} a_{ij}^2\big)^{1/2}$ is the Frobenius norm
of a matrix ${\bf A}$ and $\lambda_j = \hsigma_j^{(k)} \lambda_0$.
For the notational simplicity, hereafter, we denote $f({\bf B};\hat{\bs{\sigma}}^{(r)}, \lambda_0)$ and $\widehat{\bf B}^{(r+1)}$ as $f({\bf B})$ and $\widehat{\bf B}$, respectively. 
As $f({\bf B})$ is the sum of the smooth function of ${\bf B}$ (the
square of the Frobenius norm) and non-smooth convex functions,
$f({\bf B})$ satisfies the conditions of Theorem 4.1 in Tseng (2001).
Thus, the iterative sequence $\{\hat{\bf B}^{(r)}\}$ in the CD algorithm with cyclic rule converges to the stationary point of $f({\bf B})$,
where $\hat{\bf B}^{(r)}$ denotes the iterative solution
of the CD algorithm at the $r$th iteration, and
each iteration is counted when one cycle is finished
(i.e., $\beta_{12},\ldots,\beta_{p-1,p}$ have been updated.).

To develop the parallel CD (PCD) algorithm,
we consider a row-wise update for $\hat{\bf B}$, which is
one of the possible orderings in the cyclic rule.
The main idea of the parallel CD algorithm is that
the $p$ Lasso subproblems are independent in the sense that 
$\beta_{j,k}$ does not need information $(\beta_{j,l})$ for $l\neq k$. 
To be specific, 
let $\hbetabs^{(j), [cur]} =(\hbeta_{j,1}^{[cur]},\ldots,\hbeta_{j,j-1}^{[cur]}, 0, \hbeta_{j,j+1}^{[cur]}, \ldots,  \hbeta_{j,p}^{[cur]})$ and $\hbetabs^{(j), [next]}$ be the $j$th
row of the current and next iterative solutions of $\hat{\bf B}$, 
respectively. The following Proposition \ref{update} shows that the row of $\hat{\bf B}$ can be updated in parallel.

\begin{proposition} \label{update}
Let ${\bf E} = ({\bf e}_1,\ldots,{\bf e}_p)$ be
a current residual matrix defined with 
 ${\bf e}_k
= \bx_k - \bX \hbetabs_{-k}^{[cur]}$, 
and let $\hat{\bf B}^{[cur]}$
and $\widehat{\bf B}^{[next]}$ be
the current and next iterative solution for the coefficient
of the joint Lasso subproblem, respectively.
Suppose we update the rows of $\hat{\bf B}^{[next]}$
from the first row to the last row.
Then, each row  $\hbetabs^{(j),[next]}$ of $\hat{\bf B}^{[next]}$
for $j=1,\ldots,p$ 
can simultaneously be updated by following updating equations: 
\[
{\bf a}_j =(a_{jk})_{k=1}^p = \bx_j^T {\bf E}/n + \hbetabs^{(j),[cur]},~
{a}_{jj} \gets 0,~
\hbetabs^{(j),[next]} = {\bf S}_{\lambda_1,\ldots,\lambda_p}({\bf a}_j),~
{\bf E} \gets {\bf E} + \bx_j (\hbetabs^{(j),[cur]}-\hbetabs^{(j),[next]}),
\] 
where ${\bf S}_{\lambda_1,\ldots,\lambda_p}({\bf x})
= (\mbox{\rm Soft}_{\lambda_j}(x_j))_{1\le j \le p}$,
$\mbox{\rm Soft}_\lambda(x) = \mbox{\rm sign}(x)(|x|-\lambda)_+$, and
$(x)_+ = \max(x,0)$.
\end{proposition}

\begin{proof}
As described in Section 3.1, the CD algorithm updates
each element  $\hbeta_{j,k}^{[next]}$ in $\hbetabs^{(j),[next]}$ \
by  
\[
{\bf e}_k = \bx_k - \sum_{l<k} \bx_l \hat{\bs{\beta}}_{-l}^{[next]}
- \sum_{l>k} \bx_l \hat{\bs{\beta}}_{-l}^{[cur]}, \quad
a_{jk} = \bx_j^T {\bf e}_k / n + \hat{\beta}_{j,k}^{[cur]}, \quad
\hat{\beta}_{j,k}^{[next]} = \mbox{Soft}_\lambda(a_{jk}).
\] 
Consider updating the first row of $\hat{\bf B}^{(next)}$ by
the CD algorithm. Then, for $k=2,\ldots,p$, the above updating equations becomes 
\[
{\bf e}_k = \bx_k
- \sum_{l=2}^p \bx_l \hat{\beta}_{-l}^{[cur]}, \quad
a_{1k} = \bx_1^T {\bf e}_k / n + \hat{\beta}_{1,k}^{[cur]}, \quad
\hat{\beta}_{1,k}^{[next]} = \mbox{Soft}_\lambda(a_{1k}).
\] 
The update for $\hat{\beta}_{1,k}^{[next]}$
only needs information on the current residual vector  
${\bf e}_k$
and $\hbeta_{1,k}^{[cur]}$.   Hence,
the updates of $\hat{\beta}_{1,k}^{[next]}$
and $\hat{\beta}_{1,l}^{[next]}$ for $k\neq l$ are independent
in the sense that $\hbeta_{1,k}^{[next]}$ does not
depend on $\hbeta_{1,k}^{[next]}$, and vice versa.
Combining these equations for $j=2,\ldots,p$,
we can represent $(p-1)$ updating equations with the following
vector form:
\[
(0, \hbeta_{12}^{[next]},\ldots,\hbeta_{1p}^{[next]})
= (0, \mbox{Soft}_{\lambda_2}(a_{12}), \ldots, \mbox{Soft}_{\lambda_p}(a_{1p})),
\]
where $\lambda_j = \hsigma_j \lambda_0$ and 
${\bf a}_1^T = (0, {\bf e}_2^T \bx_1/n, \ldots, {\bf e}_p^T \bx_1/n)
+ (0, \hbeta_{1,2}^{[cur]},\ldots, \hbeta_{1,p}^{[cur]})
=\bx_1^T {\bf E}/n + \hbetabs^{(1),[cur]} - (\bx_1^T {\bf e}_1/n) {\bf i}_1$,
where $\hbeta_{1,1}^{[cur]}=0$ and ${\bf i}_j$ is the $j$th row of $p$-dimensional
identity matrix. 
After updating $\hbetabs^{(1),[next]}$,
the current residual matrix is also updated by 
${\bf E} \gets {\bf E} + \bX_1 (\hbetabs^{(1),[cur]}-
\hbetabs^{(1),[next]})$. 
Then, the updated ${\bf e}_k$ becomes $\bx_k - \bx_1 \hbeta_{1,k}^{[next]}
- \sum_{l=2}^p \bx_l \hbeta_{l,k}^{[cur]}$. 
Using these equations, we can express a general form of updating
equations as 
\[
{\bf a}_j =(a_{jk})_{k=1}^p = \bx_j^T {\bf E}/n + \hbetabs^{(j),[cur]},~
{a}_{jj} \gets 0,~
\hbetabs^{(j),[next]} = {\bf S}_{\lambda_1,\ldots,\lambda_p}({\bf a}_j),~
{\bf E} \gets {\bf E} + \bx_j (\hbetabs^{(j),[cur]}-\hbetabs^{(j),[next]}),
\] 
where, for computational simplicity, we
calculate 
${\bf a}_j$ by $\bx_j^T {\bf E}/n + \hbetabs^{(j),[cur]}$ and then set $a_{jj} = 0$.  
Applying this sequence
of updating equations from the first row ($j=1$) to the last row ($j=p$) of $\widehat{\bf B}^{[next]}$ is equivalent to the cyclic CD algorithm
for the joint Lasso subproblem.
\end{proof}

In Proposition \ref{update}, the row-wise updating equations consist of basic linear algebra
operations such as matrix-matrix multiplication and
element-wise soft-thresholding, which are adequate for parallel computation using GPUs. To fully utilize the GPUs,
we use the \texttt{cuBLAS} library for linear algebra operations
and develop CUDA kernel functions for the element-wise soft-thresholding, parallel update for  $\hat{\bs{\sigma}}$ , and symmetrization of  $\hat{\bs{\Omega}}$.   We refer this CD algorithm to the parallel CD (PCD) algorithm in the sense of that $p$ elements in each row of $\widehat{\bf B}$
are simultaneously updated if $p$ GPUs (i.e., $p$ CUDA cores) are available.
Note that  $\hbeta_{jj}$ for $j=1,\ldots,p$   are fixed 
with $0$ in the PCD algorithm, which is handled explicitly in the implementation.
The whole PCD algorithm with warm start strategy for the SPMESL is summarized in Algorithm \ref{alg:pcd_spm}. In Algorithm \ref{alg:pcd_spm}, we use a convergence criterion $\|{\bf B}^{(r+1)} - {\bf B}^{(r)}\|_\infty < \delta$ to check the convergence of ${\bf B}^{(r)}$, which is different
to the convergence criterion $\|\bs{\beta}_{-j}^{(r+1)} - \bs{\beta}_{-j}^{(r)}\|_{\infty} < \delta$ in the CD algorithm.  To ensure that the CD and the PCD algorithm provide the same solution, we show that the PCD algorithm obtains a solution that is sufficiently close to the solution of the CD algorithm if two algorithms use the same initial values in the following Theorem 1:
\begin{theorem}
For a given vector $(\hsigma_1, \ldots, \hsigma_p)$, a tuning parameter $\lambda_0$,
and a convergence tolerance $\delta>0$, 
let  ${\bs{\beta}}_{-k}^{(r)}$ and ${\bf b}_k^{(r)}$  be iterative solutions at the $r$th iteration by the CD algorithm with a convergence criterion  $\|\bs{\beta}_{-k}^{(r+1)} - \bs{\beta}_{-k}^{(r)}\|_{\infty} < \delta$  and the PCD algorithm with a convergence criterion  $\|{\bf B}^{(r+1)} - {\bf B}^{(r)}\|_\infty < \delta$,   respectively. Let $\hat{\bs{\beta}}_{-j}$ and $\hat{\bf b}_j$ be the solutions of the CD and the PCD that satisfy the given convergence criteria. Suppose that the CD and the PCD algorithms use the same initial point  ($\hat{\bs{\beta}}_{-k}^{(0)}=\hat{\bf b}_k^{(0)}, 1\le j \le p)$. Then, $\|\hat{\bf b}_k - \hat{\bs{\beta}}_{-k}\|_\infty$   is also bounded by $\delta$.
\end{theorem}
\begin{proof} 
As the function $g(\bs{\beta}_{-k};\hsigma_k, \lambda_0)$ is convex with respect to 
$\bs{\beta}_{-k}$, it is easy to show that the coordinate-wise minimization of $g(\bs{\beta}_{-k};\hsigma_k, \lambda_0)$ satisfies the conditions (B1)--(B3) and (C1) in \cite{Tseng2001}. To see this, let $\eta = (\eta_1,\ldots,\eta_{p-1})^T=(\beta_{1k},\ldots,\beta_{k-1,k},\beta_{k+1,k},\ldots,\beta_{pk})^T$, $\tau = (\tau_1,\ldots,\tau_{p-1})^T$, and
$\tau_m = \lambda_{m}$ for $1\le m \le k-1$ and $\tau_m=\lambda_{m+1}$ for $k\le m\le p-1$.
We further let $f_0(\eta) = \frac{1}{2n}\|\bx_j - \bX_{-j} \eta\|_2^2$ and $f_m(\eta_m) = \tau_m | \eta_{m}|$ for $1\le m \le p-1$,
where $\bX_{-j} = (\bx_1,\ldots, \bx_{j-1},\bx_{j+1},\ldots, \bx_p)$. Then, we can represent $g(\bs{\beta}_{-k};\hsigma_k, \lambda_0)$ as $f(\eta)=f_0(\eta) + \sum_{m=1}^{p-1} f_m(\eta_{m})$.
With this representation, it is trivial that $f_0$ is continuous on ${\rm dom} f_0$  (B1) and $f_0, f_1, \ldots, f_{p-1}$ are lower semi-continuous (B3). As $f_0, f_1, \ldots, f_{p-1}$ are convex functions, the function $\eta_m \mapsto f(\eta_1,\ldots, \eta_{p-1})$ for each $m\in \{1,\ldots, p-1\}$ and $(\eta_l)_{l\neq m}$
is also convex and hemivariate (B2).
The function $f_0$ also satisfies that ${\rm dom} f_0$ is open and $f_0$ tends to $\infty$ at every boundary point of ${\rm dom} f_0$ because the domain of $f_0(\eta)$ is $\RR^{p-1}$ and $f_0(\eta)$ is the sum of squares of errors.
Thus, by Theorem 5.1 in \cite{Tseng2001},
the cyclic CD algorithm guarantees that $\bs{\beta}_{-k}^{(r)}$ converges to a stationary point of $g(\bs{\beta}_{-k};\hsigma_k, \lambda_0)$.
As the same updating order $(1\to 2 \to \cdots \to p)$ and equation are applied
with the same initial values in the proposed CD and PCD algorithms, the sequence $\{{\bf b}_k^{(r)}\}$ by
the PCD is equivalent to the sequence $\{ \bs{\beta}_{-k}^{(r)}\}$.
That is,  $\bs{\beta}_{-k}^{(r)}={\bf b}_k^{(r)}$ for $r\ge 0$.
Let $K_{CD}$ and $K_{PCD}$ be the iteration numbers that satisfies the convergence criteria of the CD and the PCD, respectively.
As $\| \bs{\beta}_{-k}^{(r)} - \bs{\beta}_{-k}^{(r-1)}\|_\infty=\| {\bf b}_k^{(r)} - {\bf b}_k^{(r-1)}\|_\infty \le \| {\bf B}^{(r)} - {\bf B}^{(r-1)}\|_\infty$, it is satisfied that $K_{PCD} \ge K_{CD}$.
As the $p$-dimensional Euclidean space with $L_\infty$-norm is Banach space,
the convergent sequence $\{ \bs{\beta}_{-k}^{(r)}\}$ and $\{ {\bf b}_{k}^{(r)}\}$ is
a Cauchy sequence.
Thus, from the definition of the Cauchy sequence, 
for a given $\delta>0$, there exists $K$ such that
$\| \bs{\beta}_{-k}^{(u)} - \bs{\beta}_{-k}^{(v)}\|_\infty < \delta$ for $u,v \ge K$.
Take $K=K_{CD}$, $u=K_{CD}$, and $v=K_{PCD}$. Then, $\bs{\beta}_{-k}^{(K_{CD})} = \hat{\bs{\beta}}_{-k}$
and $\bs{\beta}_{-k}^{(K_{PCD})} = \hat{\bf b}_k$.
Hence, the $L_\infty$-norm of the difference of $\hat{\bs{\beta}}_{-k}$ and
$\hat{\bf b}_k$ is bounded by $\delta$.
\end{proof}
For convergence of $\sigma_j$, we also use the same convergence criterion 
$|\sigma_j^{(r)}- \sigma_j^{(r-1)}|<\delta$ as in the CD algorithm.
To reduce the computational costs in the PCD algorithm, at each iteration, we remove some 
columns of ${\bf B}$ in the problem if the corresponding $\sigma_j$ satisfies the 
convergence criterion $|\sigma_j^{(r)}- \sigma_j^{(r-1)}|<\delta$.
This additional procedure needs the rearrangement of the coefficient matrix ${\bf B}$. In the implementation, we use an index vector and a convergence flag vector to
implement the additional rearrangement procedure efficiently.
Thus, the  PCD algorithm requires more computational costs than
the CD algorithm for the SPMESL as the CD algorithm runs consequently for $j=1,\ldots, p$ and does not need the rearrangement procedure.  
In the next section, however, we numerically show that the PCD algorithm becomes more
efficient compared to the CD algorithm when either the number of variables
or the sample size increases, although the PCD has more computational costs.

\begin{algorithm}
\caption{Parallel CD algorithm with warm start strategy for the SPMESL}\label{alg:pcd_spm}
\begin{algorithmic}[1]
\Require ${\bf X}$, $\lambda_0$, $\hat{\sigma}^{(0)}=1$,
${\bf B}^{(0)} = (\hbeta_{ij}^{(0)})=(\hat{\bs{\beta}}_{-1}^{(0)}, \ldots, \hat{\bs{\beta}}_{-p}^{(0)})={\bf 0}$,
convergence tolerance $\delta$.
\State Set $n_c=p$, $I=(1,\ldots,p)$, and $F = (1,\ldots,1)$ 
\Repeat{~~$r=0, 1, 2 \ldots$}  \smallskip 
\State $\lambda \gets (\hat{\sigma}_{I_1}^{(r)}\lambda_0,\ldots,\hat{\sigma}_{I_{n_c}}^{(r)}\lambda_0)$ \Comment{Initialization of joint lasso subproblem}  \smallskip
\State $\hat{\bf B}^{[cur]} \gets \hat{\bf B}^{(r)}, \hat{\bf B}^{[next]} \gets \hat{\bf B}^{[cur]}$ \Comment{Warm start strategy} 
 \smallskip
\State ${\bf E} = (\bx_{I_1},\ldots, \bx_{I_{n_c}})- \bX \hat{\bf B}^{[cur]}$  \smallskip  
\Repeat{~~$m = 0,1,2,\ldots$} 
 \smallskip
\State $\widehat{\bf B}^{[cur]} \gets \widehat{\bf B}^{[next]}$
 \smallskip
\For{$j = 1,\cdots, p$}
 \smallskip 
\State ${\bf a} = \bx_j^T{\bf E}/n + \hbetabs^{(j),[cur]}$  \smallskip 
\State $a_{j} \gets 0$  \smallskip 
\State $\hbetabs^{(j),[next]} = {\bf S}_{\lambda_1,\ldots,\lambda_p}({\bf a})$
 \smallskip 
\State ${\bf E} = {\bf E} + \bx_i (\hbetabs^{(j),[cur]}-
\hbetabs^{(j),[next]})$  \smallskip 
\EndFor
 \smallskip
\Until{ $\|\hat{\bf B}^{[next]} - \hat{\bf B}^{[cur]}\|_\infty, < \delta$}
\Comment{End of joint lasso subproblem}  \smallskip
\State $\hat{\bf B}^{(r+1)} \gets \hat{\bf B}^{[next]}$
 \smallskip
\State Update $\displaystyle \hat{\sigma}_{I_j}^{(r+1)} =\frac{\|\bx_{I_j} - \bX \hbetabs_{-I_j}^{(r+1)}\|_2}{\sqrt{n}}$ in parallel  \smallskip
\State Calculate $F_j = I(|\hat{\sigma}_j^{(r+1)}-\hat{\sigma}_j^{(r)}| \ge \delta)$ in parallel  \smallskip
\State Set $l=0$ \smallskip
\For{$j=1, \cdots, n_c$} \smallskip
	\State {\bf if}~~$F_j=1$~~{\bf then} \smallskip \newline
	\hspace*{4.5em} $l \gets l + 1$, $I_l \gets j$ \smallskip \newline
	\hspace*{3em}{\bf end if}
\EndFor
\State Set $n_c = l$, $\hat{\bf B}^{[next]} \gets (\hat{\bs{\beta}}_{-I_1}^{[next]},\ldots, \hat{\bs{\beta}}_{-I_{n_c}}^{[next]})$  
\Until{ $n_c = 0$}  \smallskip
\State Calculate $\homega_{jj} = \hsigma_j^{-2}$ and $\homega_{jk} = - \hbeta_{jk} \hsigma_k^{-2}$ in parallel \Comment{Initial estimate
for $\Omega$}  \smallskip
\State
 \smallskip Update $\homega_{jk}$ in parallel \Comment{Symmetrization}  \smallskip \newline
\hspace*{1.5em}{\bf if}~~{$|\homega_{jk}| > |\homega_{kj}|$} {\bf then}  \smallskip \newline
\hspace*{3em}$\homega_{jk} \gets \homega_{kj}$  \smallskip \newline
\hspace*{1.5em}{\bf else} \newline
\hspace*{3em}$\homega_{kj} \gets \homega_{jk}$  \smallskip \newline
 \smallskip
\hspace*{1.5em}{\bf end if}
\Ensure $\hat{\bs{\Omega}} = (\homega_{jk})$
\end{algorithmic}
\end{algorithm}

\section{Numerical Study}

\subsection{Data construction and simulation settings}

In this section, we numerically investigate the computational
efficiency of the proposed CD and PCD algorithms
and the estimation performance of the SPMESL
with comparisons to other existing methods.
To proceed the comparison on various circumstances, 
we first consider four network structures for the precision matrix defined as follows:

\begin{itemize}
\item AR(1): AR(1) network is also known as a chain graph.
We define a precision matrix $\Omega$ for AR(1) as
\[
\Omega = (\omega_{ij})_{1\le i,j \le p}
= \left\{
\begin{array}{ll}
1 & \mbox{if } i = j \\
0.48 & \mbox{if } |i-j| = 1\\
0 & \mbox{otherwise}
\end{array} \right.
\]

\item AR(4): In AR(4) network, each node is connected to neighborhood nodes
whose distance is less than or equal to 4,
where the distance of two nodes $i$ and $j$ is defined as
$d(i,j)=|i-j|$. The precision matrix $\Omega$ corresponding to AR(4) network is defined as
\[
\Omega = (\omega^{ij})_{1\le i,j \le p}
= \left\{
\begin{array}{ll}
0.6^{|i-j|} & \mbox{if } |i-j| \le 4\\
0 & \mbox{otherwise}
\end{array} \right.
\]

\item Scale-free:
Degrees of nodes follow the power-law distribution having a form
$P(k) \propto k^{-\alpha}$,
where $P(k)$ is a fraction of nodes having $k$ connections and $\alpha$ is a preferential attachment parameter.
We set $\alpha = 2.3$, which is used in \cite{Peng2009}, and we generate a scale-free network structure by using
the Barab\'asi and Albert (BA) model \citep{BA1999}.
With the generated network structure $G=(V,E)$,
we define a precision matrix corresponding to $G=(V,E)$
by following the steps applied in \citep{Peng2009}:

\noindent
$
(i)~ \tilde{\Omega} = (\tilde{\omega}_{ij})_{1\le i,j \le p}
= \left\{
\begin{array}{ll}
1 & \mbox{if } i = j \\
U & \mbox{if } (i,j) \in E\\
0 & \mbox{otherwise}
\end{array} \right.,
$ $ U \sim \mathcal{U}([-1,-0.5]\cup[0.5,1])$.
 \smallskip\\
{$(ii)$ $\Omega =({\omega}_{ij})_{1\le i,j \le p}
=\frac{\tilde{\omega}_{ij}}{1.5 \sum_{k\neq i} |\tilde{\omega}^{ik}|}
$ \hspace{1cm}}
 \smallskip\\
{$(iii)$ $\Omega = (\Omega + \Omega^T)\big/2$ \hspace{2.4cm}}
 \smallskip\\
{$(iv)$ $\omega_{ii} = 1$ for $i=1,2,\ldots,p$ \hspace{1.5cm}}

\vspace{2mm}

\item Hub:
Following \citep{Peng2009}, for $p=100$, a hub network consists of 10 hub nodes whose degrees are around 15 and 90 non-hub nodes whose degrees lie between 1 and 3. Edges in the hub network are randomly selected with the above conditions. With the given network structure $G=(V,E)$,
we generate a precision matrix by the procedure described in
the scale-free network generation.

\end{itemize}

\noindent To avoid nonzero elements having considerably small magnitudes (i.e., the absolute value of element),
we generate $(p/100)$ subnetworks, each of which consists of 100 nodes,
and set nonzero elements having magnitudes less than 0.1 to 0.1 (i.e., ${\rm sign}(\omega_{ij})/10$)
for the scale-free and hub networks.
For example, we generate five subnetworks having 100 nodes when $p=500$
for the scale-free and hub networks.
We depict the generated four network structures in Figure \ref{fig:net}
for $p=500$. 
 Among four networks, AR(1) and AR(4) networks correspond the circumstances that the variables are measured in a specific order, and
the hub and scale-free networks are frequently observed in real-world problems such as the gene regulatory networks and functional brain networks.

\begin{figure}[!htb]
\begin{minipage}[l]{.45\linewidth}
\centerline{\epsfig{file=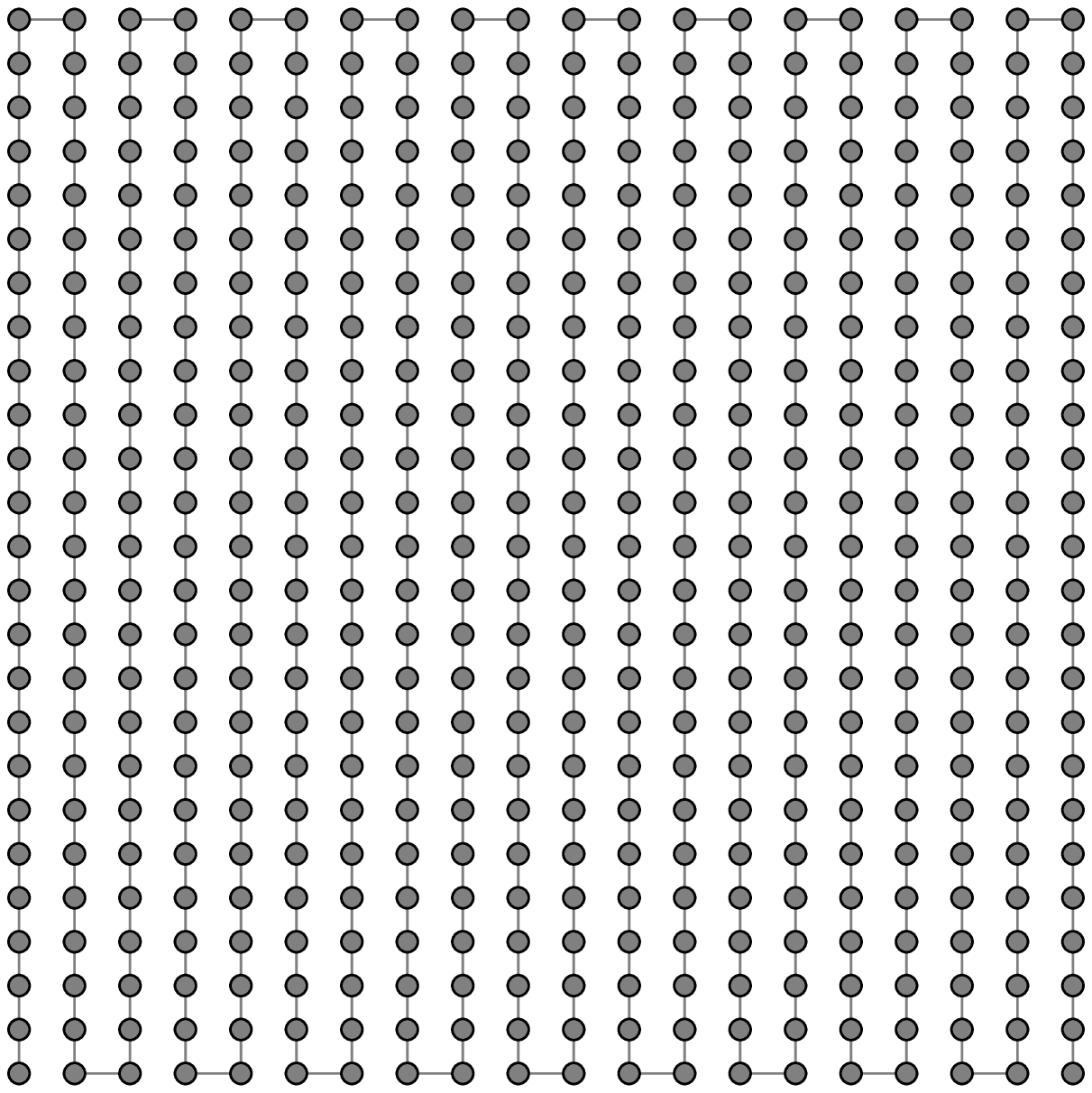,width=1\textwidth,height=0.32\textheight}}
\centerline{(a) AR(1)}
\end{minipage}
\begin{minipage}[l]{.45\linewidth}
\centerline{\epsfig{file=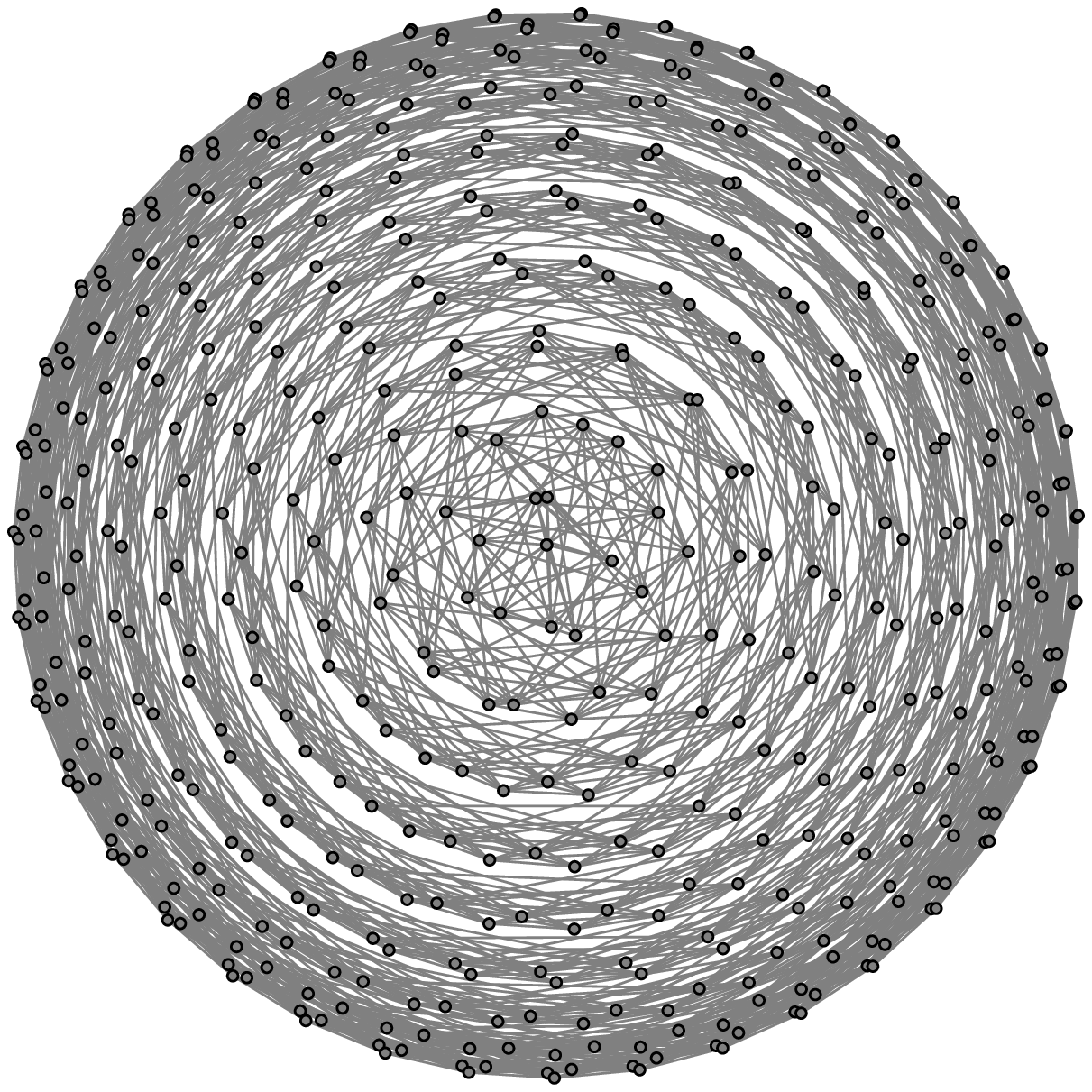,width=1\textwidth,height=0.32\textheight}}
\centerline{(b) AR(4)}
\end{minipage}
\begin{minipage}[c]{.45\linewidth}
\centerline{\epsfig{file=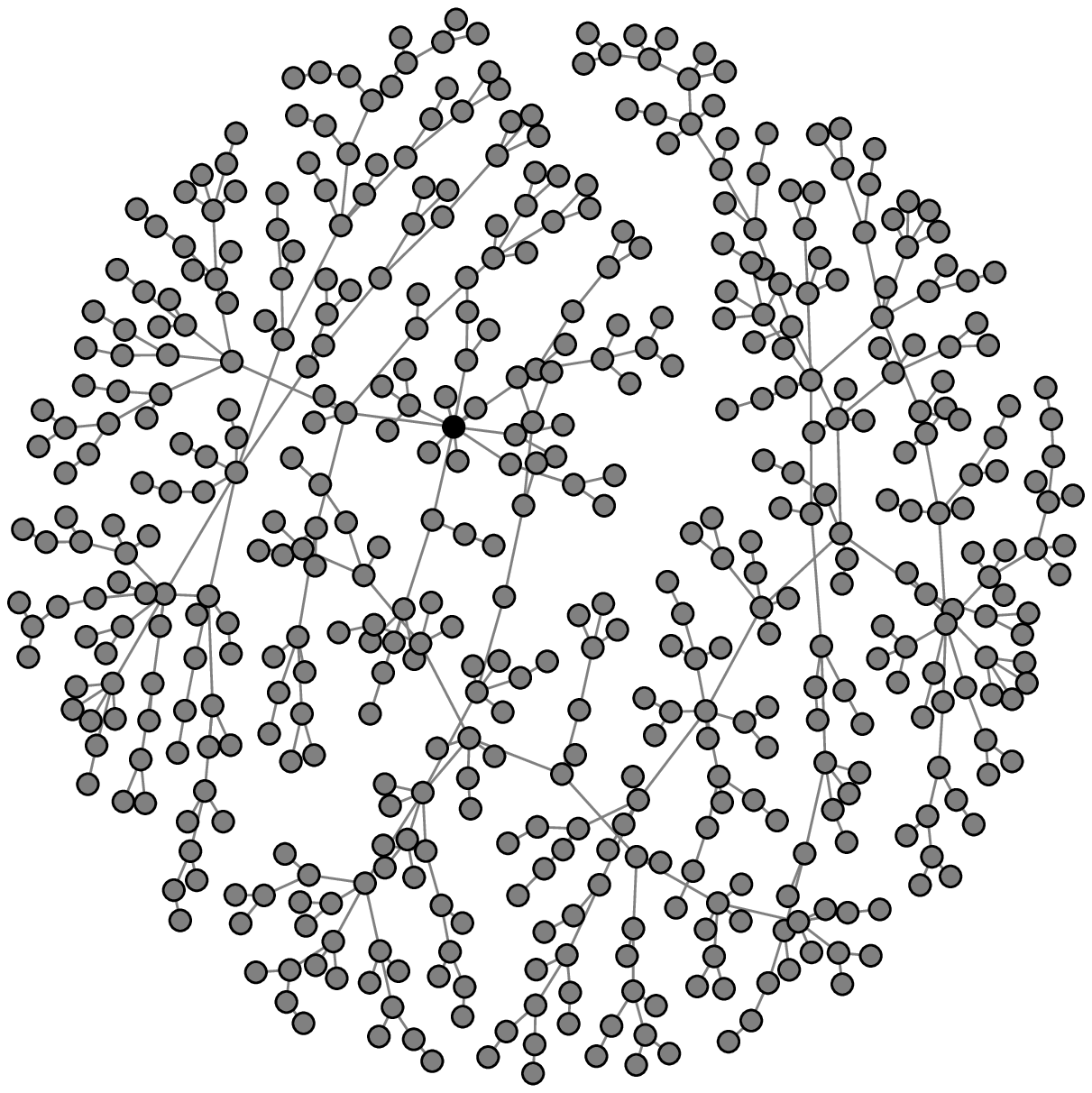,width=1\textwidth,height=0.32\textheight}}
\centerline{(c) SF}
\end{minipage}
\hspace{10mm}
\begin{minipage}[c]{.45\linewidth}
\centerline{\epsfig{file=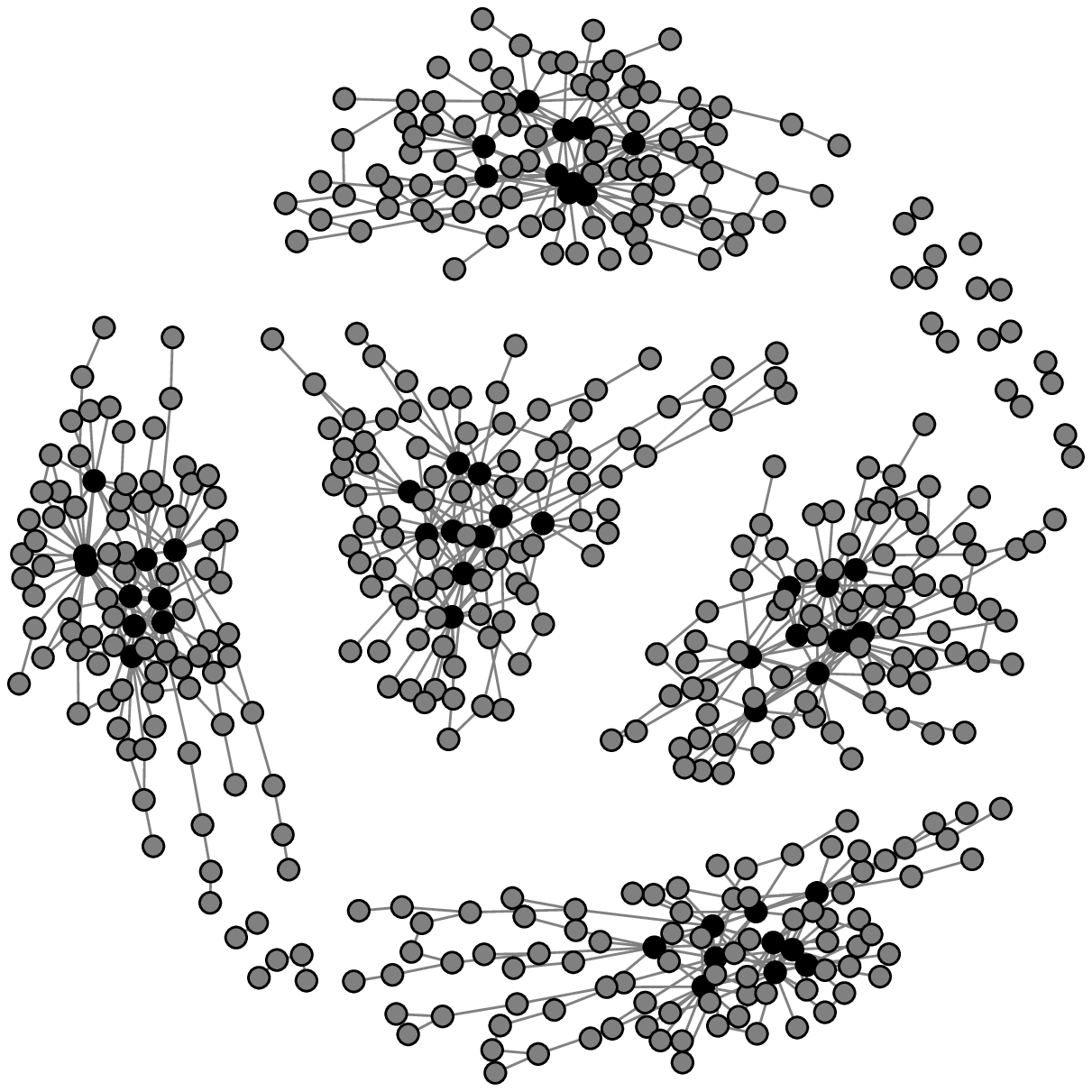,width=1\textwidth,height=0.32\textheight}}
\centerline{(d) Hub}
\end{minipage} 
\caption{Four network structures of precision matrix:
(a) AR(1), (b) AR(4), (c) Scale-free (SF), and (d) Hub-network (Hub).
Nodes in black denote nodes whose degrees are more than 9.} \label{fig:net}
\end{figure}

For comparison of the computational efficiency, 
we consider the number of variables $p=1000$, and sample size $n=250, 500$. To provide a benchmark of the computational efficiency,
we also consider the well-known existing methods
such as the CLIME \citep{Cai2011b}, graphical lasso (GLASSO) \citep{Friedman2008}, and convex partial correlation estimation method (CONCORD) \citep{Khare2015}.
As the computation times of the algorithms are affected by the level of sparsity
on the estimate of the precision matrix,
we set $\lambda_0 = \sqrt{4n^{-1}\log p}$ for LARS-CPU (SPMESL-LARS)
and CD-CPU (SPMESL-CD). Further, we search the tuning parameters for other methods
to obtain the level of sparsity similar to the one obtained by the SPMESL.
For example, for AR(1) and AR(4) networks with $p=1000$, the averages of the estimated edges of all methods are around 1000, which is $0.2$\% of
the total possible edges. 
With the chosen penalty levels,
the computation time for each method is
measured in CPU time (seconds) by using
a workstation (Intel(R) Xeon(R) W-2175 CPU (base: 2.50 GHz, max-turbo: 4.30 GHz) and 128 GB RAM)
with NVIDIA GeForce GTX 1080 Ti.
We measure the average computation times and standard errors
over 10 datasets.

For comparison of the estimation performance,
we consider evaluating the estimation performance of the SPMESL with $\lambda_{pb}=\sqrt{2} L_n (k/p)$, $\lambda_{univ}=\sqrt{2 n^{-1} \log (p-1)}$, and $\lambda_{ub}=\sqrt{4 n^{-1} \log p}$ as there are three different suggestions for $\lambda_0$ without a comparison of the estimation performance, where $k$ is a real solution of $k=L_1^4(k/p) + 2 L_1^2(k/p)$,
$L_n(t) = n^{-1/2} \Phi^{-1}(1-t)$, and $\Phi^{-1}(t)$ is the standard normal quantile function.
Here, we denote the SPMESL with $\lambda_{pb}$, $\lambda_{univ}$, $\lambda_{ub}$ as SPMESL-P, SPMESL-2, SPMESL-4, respectively.
As described in the computational efficiency comparison,
we apply  the three existing methods (CLIME, GLASSO, CONCORD)
to provide a benchmark of the estimation performance.
Hereafter, we referred these three methods to the \emph{tuning-search} methods.
To measure the estimation performance, we consider five performance measures-- sensitivity (SEN), specificity (SPE), false discovery rate (FDR),
miss-classification error rate (MISR), and Matthew's correlation coefficients (MCC)--for identifying the true edges and the Frobenius norm of $\bs{\Omega}^0 - \hat{\bs{\Omega}}$ for estimation error, where $\bs{\Omega}^0$ and $\hat{\bs{\Omega}}$
denote the true precision matrix and the estimate of the precision matrix, respectively.
The five performance measures for identification of the true edges
are defined as follows:
\begin{equation} \nonumber
\begin{array}{l} \medskip
{\rm SEN} \equiv \displaystyle \frac{\rm TP}{({\rm TP}+{\rm FN})},
{\rm SPE} \equiv \displaystyle \frac{{\rm TN}}{({\rm TN}+{\rm FP})},
{\rm FDR} \equiv \displaystyle \frac{{\rm FP}}{({\rm TP} + {\rm FP})},\\ \medskip
{\rm MISR} \equiv \displaystyle \frac{({\rm FP}+{\rm FN})}{p(p-1)/2},
~{\rm MCC}(\hat{\Omega},\Omega^0) \equiv
\frac{{\rm TP} \times {\rm TN} - {\rm FP} \times {\rm FN}}{\sqrt{({\rm TP}+{\rm FP})({\rm TP}+{\rm FN})({\rm TN}+{\rm FP})({\rm TN}+{\rm FN})}},
\end{array}
\end{equation}
where ${\rm TP} = \sum_{j<k} I(\hat{\omega}_{jk}\neq 0, \omega^0_{jk}\neq 0)$,
${\rm FP} = \sum_{j<k} I(\hat{\omega}_{jk}\neq 0, \omega^0_{jk}= 0)$,
${\rm TN} = \sum_{j<k} I(\hat{\omega}_{jk}= 0, \omega^0_{jk}= 0)$, and
${\rm FN} = \sum_{j<k} I(\hat{\omega}_{jk}= 0, \omega^0_{jk} \neq 0)$,

In the comparison of the estimation performance, we set the number of variables and sample size as 500 and 250, respectively. We generate samples
$\bX^1, \bX^2,\ldots, \bX^n \sim N({\bf 0},\bs{\Sigma})$,
where $\bs{\Sigma} = \bs{\Omega}^{-1}$ and $\bs{\Omega}$ is from a given network structure.
Unlike the comparison of computational efficiency, we need to choose criteria
for the selection of the optimal tuning parameter of the CLIME, GLASSO, and CONCORD.
For a fair comparison, we adopt the Bayesian information criterion (BIC), which is widely used in model selection, and search the optimal tuning parameter over an equally spaced grid $(0.10, 0.12, \ldots, 0.88, 0.90)$.
We generate 50 datasets, each of which we search the optimal tuning parameters for the CLIME, GLASSO, and CONCORD and apply $\lambda_{pb}$, $\lambda_{univ}$, $\lambda_{ub}$ for the SPMESL.

In this numerical study, we implement R package \texttt{cdscalreg} for the CD algorithm
with a warm start strategy for
the scaled Lasso and SPMESL, which is available at
 \url{https://sites.google.com/view/seunghwan-lee/software}.
For the other methods in this study,  we use R packages \texttt{fastclime}
for FASTCLIME, \texttt{glasso} for GLASSO, \texttt{gconcord} for CONCORD,
\texttt{scalreg} for the original algorithm for SPMESL.
Note that 
consider FASTCLIME \citep{Pang2014} for the CLIME, which obtains the CLIME estimator
efficiently and provides a solution path of the CLIME applying the parametric simplex method.
It is also worth noting that the GPUs we used are more efficient for conducting operations with single-precision values than double-precision values, but
the R programming language only supports the double-precision that
makes the efficiency of the GPU-parallel computation decrease.
Although this PCD implementation could decrease its computational
efficiency, we develop an R function for the PCD algorithm
to provide an efficient and convenient tool for R users.
If readers want to utilize GPU-parallel computation maximally, PyCUDA \citep{Klockner2012} is one of the convenient and favorable ways to implement CUDA GPU-parallel computation.

\subsection{Comparison results for computational efficiency}

Table \ref{tb:comp1} reports the average computation times and standard errors
 over 10 datasets.
From Table \ref{tb:comp1}, we numerically verify that the proposed algorithm based on the CD with warm-start strategy is more efficient than the original algorithm based on the LARS,
where the proposed algorithm (CD-CPU) is 110.2 and 466.3 times faster than LARS-CPU for
the worst case (AR(1), $n=250$) and the best case (AR(4), $n=500$), respectively.
For comparison of efficiency with other methods, overall, the CD-CPU is faster than FASTCLIME and slower than GLASSO and CONCORD. GLASSO is the most efficient algorithm in our numerical study. Its efficiency is from the sub-procedure that reduces the computational
cost by the pre-identification procedure for nonzero block diagonals of the estimate that rearranges the order of variables for a given tuning parameter described in \cite{Witten2011}.
The CONCORD is faster than the CD-CPU in general because the CD algorithm for the CONCORD is applied to minimize its objective function directly. 
However, the CD algorithm in the CD-CPU is repeatedly applied to solve the lasso subproblems. 
Even though the GLASSO and CONCORD are faster than the CD-CPU, 
the GLASSO and CONCORD are tuning-search methods while 
the CD-CPU is not. Thus, the CD-CPU becomes the most efficient
when the GLASSO and CONCORD need to evaluate more than
five tuning parameters.
For the efficiency of the PCD-GPU,
we can see that the PCD-GPU is faster than CD-CPU for
all cases except for the case of (Hub, $n=250$).
In addition, Table 1 also shows that the efficiency of the PCD-GPU increases
as the sample size increases. For example, 
all computation times of CD-CPU significantly increase when the sample size
increases from 250 to 500. However, there is no significant difference on
the computation times of PCD-GPU between the sample sizes 250 and 500.
This might show that the GPU device has idle processing units when $n=250$.
Note that the estimator of the SPMESL can be obtained by solving $p$ scaled Lasso
problems independently on multi CPU cores instead of on GPUs. However, the cost per core of CPU is more
expensive than that of GPU.
Moreover, the average computation times of the parallel computation
of the CD-CPU with 16 CPU cores with the R package \texttt{doParallel} (PCD-MPI in Table 1)
are around 20 seconds, which are worse than CD-CPU.
This inefficiency might be from the communication cost and the number of
CPU cores not enough for large $p$.

To verify the efficiency of PCD-GPU compared to CD-CPU, we conduct additional numerical studies
for CD-CPU and PCD-GPU
with $p=500, 1000, 2000, 5000$ and $n=250, 500, 1000$.
Table \ref{tb:comp2} reports
the average computation times and standard errors measured in CPU time (seconds)
for CD-CPU and PCD-GPU.
As shown in Table \ref{tb:comp2}, PCD-GPU becomes more efficient than
CD-CPU when either the number of variables or the sample size increases.
For example, CD-CPU is 1.05$\sim$1.94 times faster than PCD-GPU only 
for the Hub-network cases of $(p,n)=(500,250), (500,500), (1000,250)$; 
however, PCD-GPU outperforms CD-CPU for all the other cases and 4.71$\sim$11.65 times faster than CD-CPU when $p=5000$.
In Table \ref{tb:comp2}, we also find an advantage of the GPU-parallel computation.
Originally, the parallel computation in PCD-GPU applied to reduce the computational cost
depends on the number of variables. However, the additional numerical studies support
that the parallel computation in PCD-GPU also reduces the computation times when
the number of samples increases for a fixed $p$. This advantage is from the efficiency
of the GPU-parallel computation for the matrix-matrix and matrix-vector multiplications.
Thus, the additional numerical studies show that PCD-GPU is favorable
for the cases where either $p$ or $n$ is sufficiently large.

\begin{table}[!htb]
\begin{centering} 
\caption{The averages of the
computation times (sec.) over 10 datasets.
Numbers in the parentheses denote the standard errors.}     \label{tb:comp1}
\scalebox{0.9}{
\begin{tabular}{|c|c|c|c|c|c|c|c|c|c|} \hline
Network	&	$p$	&	$n$	&	\texttt{scalreg}	&	CD-CPU	&	PCD-GPU	 & PCD-MPI &	FASTCLIME	&	GLASSO	&	CONCORD \\	\hline
\multirow{4}{*}{AR(1)}	&	\multirow{4}{*}{1000}	&	\multirow{2}{*}{250}	&	937.7358&  8.5027&  4.4356&  18.8599  & 241.0991&  0.7091&  6.4348 	\\
	&		&		&	(4.4662)&(0.0519)&(0.0749)&   (0.7913)  & (4.0303)&(0.0081)&(0.1541)	\\
\cline{3-10}
	&		&	\multirow{2}{*}{500}	&	4345.7155&  13.9459&   4.2743& 21.4523 &  227.3807&   1.2521&  12.4249 	\\	
	&		&		&	(3.9313)&(0.0335)&(0.0397)&  (0.8523)  & (2.2561)&(0.1210)&(0.0830)	\\
\hline
\multirow{4}{*}{AR(4)}	&	\multirow{4}{*}{1000}	&	\multirow{2}{*}{250}	&	803.3684&  5.9992&  1.7896& 19.2027  &233.5897&  0.8030&  3.4643 \\
	&		&		&	(0.9841)&(0.0158)&(0.0274)& (1.0094) & (1.9512)&(0.0460)&(0.0410)	\\
	\cline{3-10}
	&		&	\multirow{2}{*}{500}	&	4043.0187&   8.6698&   2.2892& 20.0631 & 259.5524&   8.9721&   7.5435 	\\	
	&		&		&	(4.4534)&(0.0220)&(0.0317)& (0.8708) & (2.0104)&(0.2280)&(0.2210)	\\
\hline
\multirow{4}{*}{Scale-Free}	&	\multirow{4}{*}{1000}	&	\multirow{2}{*}{250}	&790.8713&  5.1367&  3.0949& 18.8765 & 238.7474&  0.6852&  3.5295	\\
	&		&		&	(0.6122)&(0.0376)&(0.0545)& (1.3815) & (1.7989)&(0.0010)&(0.0481)	\\
\cline{3-10}
	&		&	\multirow{2}{*}{500}	&	3819.9712&   9.1574&   3.0781& 19.2749 & 246.0107&   0.7945&   7.0473	\\
	&		&		&	(8.2088)&(0.0190)&(0.0255)& (0.7189) &(1.2268)&(0.0014)&(0.0775)	\\
\hline
\multirow{4}{*}{Hub}	&	\multirow{4}{*}{1000}	&	\multirow{2}{*}{250}	&787.6731&  4.7785&  7.3326 &  18.8566 &245.2853&  0.7512&  6.1875	\\
	&		&		&	(0.8286)&(0.0247)&(0.1355)& (1.3197) &(1.6676)&(0.0099)&(0.1024)	\\
\cline{3-10}
	&		&	\multirow{2}{*}{500}	&	3769.7531&   9.2215&   7.7584& 18.7167 & 250.2303&   2.2454&  12.7924\\
	&		&		&	(23.6487)&(0.0458)&(0.1942)& (0.5727) & (2.0064)&(0.1917)&(0.2707)	\\ \hline
\end{tabular}
}
\par\end{centering}
\end{table}

\begin{table}[!htb]
\begin{centering} 
\caption{The averages of the
computation times (sec.) over 10 datasets.
Numbers in the parentheses denote the standard errors.}   \label{tb:comp2}
\scalebox{0.92}{
\begin{tabular}{|c|c|cc|cc|cc|cc|} \hline
\multirow{2}{*}{Network}	&	\multirow{2}{*}{$n$}	&	\multicolumn{2}{|c|}{$p=500$}			&	\multicolumn{2}{|c|}{$p=1000$}			&	\multicolumn{2}{|c|}{$p=2000$}			&	\multicolumn{2}{|c|}{$p=5000$}			\\	\cline{3-10}
	&		&	CD	&	PCD	&	CD	&	PCD	&	CD	&	PCD	&	CD	&	PCD	\\	\hline
\multirow{6}{*}{AR(1)}	&	\multirow{2}{*}{250}	&	2.2800 & 1.9528 & 8.7067 & 4.4149 & 38.1056 & 12.0651 & 297.7769 &  60.1581	\\
	&		&	(0.0224) & (0.0414) & (0.0268) & (0.0765) & (0.0841) & (0.1121) & (0.2836) & (0.9851)	\\	\cline{2-10}
	&	\multirow{2}{*}{500}	&	3.4519 & 1.7117 & 13.8964 &  4.2744 & 59.2341 & 12.8309 & 436.7230 &  64.4670	\\
	&		&	(0.0134) & (0.0330) & (0.0419) & (0.0445) & (0.0736) & (0.1152) & (2.7395) & (0.7545)	\\	\cline{2-10}
	&	\multirow{2}{*}{1000}	&	5.8935 & 1.7364 & 23.5556 &  4.9784 & 99.7887 & 15.8946 & 700.8538 &  85.5756	\\
	&		&	(0.0265) & (0.0217) & (0.0636) & (0.0561) & (0.1460) & (0.1447) & (1.0246) & (0.6645)	\\
	\hline
\multirow{6}{*}{AR(4)}	&	\multirow{2}{*}{250}	&	1.4538 & 0.8794 & 6.1486 & 1.7679 & 27.9208 &  4.9945 & 232.7940 &  20.9527	\\
	&		&	(0.0077) & (0.0258) & (0.0265) & (0.0271) & (0.0666) & (0.2905) & (0.2484) & (0.3062)	\\
\cline{2-10}
	&	\multirow{2}{*}{500}	&	2.1363 & 0.9950 & 8.5920 & 2.2517 & 37.8047 &  6.2053 & 305.4613 &  26.2313	\\
	&		&	(0.0116) & (0.0177) & (0.0254) & (0.0372) & (0.0777) & (0.2523) & (0.3319) & (0.6105)	\\	\cline{2-10}
	&	\multirow{2}{*}{1000}	&	6.1868 & 1.5552 & 21.5172 &  4.1366 & 80.3323 & 12.1282 & 499.1462 &  57.1606	\\
	&		&	(0.0312) & (0.0155) & (0.1308) & (0.0284) & (0.3212) & (0.0558) & (2.8728) & (0.7015)	\\	\hline
\multirow{6}{*}{Scale-Free}	&	\multirow{2}{*}{250}	&	1.2493 & 1.1620 & 4.9605 & 2.9021 & 22.1757 &  6.8921 & 186.5103 &  35.0760	\\
	&		&	(0.0067) & (0.0328) & (0.0227) & (0.0426) & (0.0362) & (0.0762) & (0.2680) & (0.6203)	\\	\cline{2-10}
	&	\multirow{2}{*}{500}	&	2.1820 & 1.1572 & 8.7524 & 2.9931 & 38.4489 &  8.3851 & 300.9288 &  44.9643	\\
	&		&	(0.0094) & (0.0109) & (0.0196) & (0.0276) & (0.0646) & (0.1162) & (0.6319) & (0.5134)	\\	\cline{2-10}
	&	\multirow{2}{*}{1000}	&	3.7781 & 1.2210 & 15.0969 &  3.6376 & 65.1003 & 11.1231 & 479.8427 &  63.4638	\\
	&		&	(0.0074) & (0.0122) & (0.0141) & (0.0298) & (0.1194) & (0.0762) & (1.4722) & (0.9511)	\\	\hline
\multirow{6}{*}{Hub}	&	\multirow{2}{*}{250}	&	1.2021 & 2.3266 & 4.5724 & 6.9508 & 19.9179 & 11.6372 & 168.9348 &  35.8963	\\
	&		&	(0.0075) & (0.0434) & (0.0207) & (0.1257) & (0.0456) & (0.2135) & (0.2214) & (0.6634)	\\\cline{2-10}
	&	\multirow{2}{*}{500}	&	2.2632 & 2.3915 & 8.9738 & 7.6092 & 38.4337 & 15.7290 & 300.2697 &  53.2606	\\
	&		&	(0.0078) & (0.0335) & (0.0220) & (0.1866) & (0.0630) & (0.3070) & (0.4684) & (0.4222)	\\	\cline{2-10}
	&	\multirow{2}{*}{1000}	&	3.9263 & 2.6366 & 15.6040 &  9.6065 & 66.8115 & 21.4761 & 492.7247 &  77.6366	\\
	&		&	(0.0108) & (0.0342) & (0.0279) & (0.1127) & (0.1204) & (0.1523) & (1.1698) & (0.4552)	\\
\hline
\end{tabular}
}
\par\end{centering}
\end{table}

\subsection{Comparison results for estimation performance}

Tables \ref{tb:prf1} and \ref{tb:prf2} report the averages of the number of estimated edges ($|\hat{E}|$) and the six performance
measures over 50 data sets for AR(1), AR(4), Scale-free, and Hub networks.
From the results in Tables \ref{tb:prf1} and \ref{tb:prf2},
we can find several interesting features.
First, focusing on comparing SPMESL-P, SPMESL-2, and SPMESL-4, the SPMESL-P obtains the smallest estimation error in Frobenius norm, and the SPMESL-4 has the largest estimation error for all network structures.
The estimation error of SPMESL-2 is located near the middle of an interval defined by the estimation errors of the SPMESL-P and SPMESL-4.
However, for the performance in the identification of the true edges,
the SPMESL-4 has the smallest SEN and FDR for all network structures, and the SPMESL-2 obtains the lowest MISR and the highest MCC for Scale-free and Hub networks
while the SPMESL-P has the highest MISR and the lowest MCC for all network structures.
For AR(1) and AR(4) network structures, the MISR and the MCC of the SPMESL-2 are worse than those of the SPMESL-4, but the differences of the SPMESL-2 and the SPMESL-4 are small. For example, the difference in the MCC of the SPMESL-2 and the SPMESL-4 are 0.0195 and 0.0095 in an original scale for AR(1) and AR(4), respectively.

Second, by comparing the SPMESL and the tuning-search methods (CLIME, GLASSO, CONCORD), the SPMESL-2 and SPMESL-4 obtain considerably small FDRs
compared to the tuning-search methods. For example, the FDRs by the tuning-search methods are over 27\% for all cases, but the FDRs by the SPMESL-2 and the SPMESL-4 are less than 6.2\%.
The SPMESL-P has the lowest MCC and the highest FDR for Scale-free and Hub networks
and obtains the second-lowest MCC and the second-highest FDR for AR(4) networks, where only the GLASSO is worse than the SPMESL-P in terms of the MCC and FDR. For AR(1) network, the SPMESL with all penalty levels are better than the tuning-search methods for identifying the true edges, and the estimation errors of the SPMESL-P and SPMESL-2 are less than those of the tuning-search methods.

Finally, we compare the estimation performance of the tuning-search methods.
For the estimation error in the Frobenius norm,
the CONCORD outperforms CLIME and GLASSO for all cases, where
the CONCORD obtains similar estimation errors to the SPMESL-P for AR(4), Scale-free, and Hub networks.
However, for the identification of the true edges,
the CLIME obtains slightly better performance than the CONCORD for all networks except the AR(1) network.
For the AR(1) network, the CONCORD outperforms CLIME and GLASSO for estimation error and identification of the true edges.
In our numerical study, the CONCORD is favorable among the three tuning-search methods we consider. Note that we adopt the BIC for three tuning-search methods for a fair comparison, but the comparison results might be changed if we apply other
model selection criteria.

From the results of the estimation performance comparison,
we recommend the SPMESL with the universal penalty level $\lambda_{univ}$
if the target problem can accept the FDR level around 5\%; the uniform-bound penalty level is only preferred when the problem only accepts the small FDR less than 1\%.
Note that we do not recommend using the probabilistic bound $\lambda_{pb}$ as the SPMESL-P has the highest FDR among three penalty levels for the SPMESL, which are over 50\% for Scale-free and Hub networks, although the SPMESL-P obtains the lowest estimation error in Frobenius norm.

\begin{table}[!htb]
\begin{centering} 
\caption{For AR(1) and AR(4) networks with $p=500$ and $n=250$, the averages of the number of estimated edges, 
the five performance measures and
the Frobenius norms of difference
of the estimate and true precision matrix
over 50 datasets.
Numbers in the parentheses denote the standard errors.}   \label{tb:prf1}
\medskip
\scalebox{0.85}{
\begin{tabular}{|c|c|c|c|c|c|c|c||c|} \hline
Network & Method & $|\hat{E}|$ & SEN & SPE & FDR & MISR & MCC & $\| \hat{\Omega} - \Omega\|_F$ \\ \hline 
   & \multirow{2}{*}{CLIME} & 897.80 & 100.00 &  99.68 &  44.33 &   0.32 &  74.48 &   9.17 \\
   &   & (5.20) & (0.00) & (0.00) & (0.31) & (0.00) & (0.21) & (0.02) \\ \cline{2-9}
   &	 \multirow{2}{*}{GLASSO} &2134.48 &  100.00 &   98.68 &   76.57 &    1.31 &   48.07 &   12.74 \\
   &   &(14.52) & (0.00) & (0.01) & (0.16) & (0.01) & (0.17) & (0.02) \\ \cline{2-9}
   &	 \multirow{2}{*}{CONCORD} & 722.16 & 100.00 &  99.82 &  30.84 &   0.18 &  83.08 &   4.96 \\
AR(1)    &   & (2.93) & (0.00) & (0.00) & (0.29) & (0.00) & (0.17) & (0.01) \\ \cline{2-9}
$(|E|=499)$   &	 \multirow{2}{*}{SPMESL-P} &649.28 & 100.00 &  99.88 &  23.12 &   0.12 &  87.62 &   3.65 \\
   &   &(1.68) & (0.00) & (0.00) & (0.20) & (0.00) & (0.11) & (0.01) \\ \cline{2-9} 
   &	 \multirow{2}{*}{SPMESL-2} &524.74 & 100.00 &  99.98 &   4.90 &   0.02 &  97.51 &   4.55 \\
   &   & (0.77) & (0.00) & (0.00) & (0.14) & (0.00) & (0.07) & (0.01) \\ \cline{2-9}
   &	 \multirow{2}{*}{SPMESL-4} & 504.40 & 100.00 & 100.00 &   1.07 &   0.00 &  99.46 &   6.39 \\
   &   & (0.36) & (0.00) & (0.00) & (0.07) & (0.00) & (0.04) & (0.01) \\ \hline 
   &	 \multirow{2}{*}{CLIME} & 908.58 &  32.55 &  99.79 &  28.70 &   1.29 &  47.65 &  22.35 \\
   &   & (2.66) & (0.08) & (0.00) & (0.15) & (0.00) & (0.08) & (0.01) \\ \cline{2-9}
   &	 \multirow{2}{*}{GLASSO} & 988.96 &  28.86 &  99.66 &  41.66 &   1.47 &  40.36 &  23.17 \\
   &   & (10.79) & (0.07) & (0.01) & (0.52) & (0.01) & (0.17) & (0.01) \\ \cline{2-9}
   &	 \multirow{2}{*}{CONCORD} & 979.40 &  33.03 &  99.74 &  32.86 &   1.33 &  46.52 &  18.46 \\
AR(4)   &   & (3.33) & (0.08) & (0.00) & (0.19) & (0.00) & (0.10) & (0.01) \\ \cline{2-9} 
$(|E|=1990)$  &	 \multirow{2}{*}{SPMESL-P} & 1206.32 &   36.31 &   99.61 &   40.09 &    1.40 &   45.99 &   18.40 \\
   &   & (2.97) & (0.07) & (0.00) & (0.16) & (0.00) & (0.09) & (0.01) \\ \cline{2-9} 
   &	 \multirow{2}{*}{SPMESL-2} & 545.30 &  25.71 &  99.97 &   6.18 &   1.21 &  48.77 &  20.57 \\
   &   & (0.94) & (0.02) & (0.00) & (0.15) & (0.00) & (0.05) & (0.01) \\ \cline{2-9}
   &	 \multirow{2}{*}{SPMESL-4} & 499.00 &  25.05 & 100.00 &   0.11 &   1.20 &  49.72 &  22.72 \\
   &   & (0.17) & (0.01) & (0.00) & (0.02) & (0.00) & (0.01) & (0.01) \\ \hline 
\end{tabular}
}
\par\end{centering}
\end{table}

\begin{table}[!htb]
\begin{centering} 
\caption{For Scale-Free and Hub networks with $p=500$ and $n=250$,
the averages of 
the number of estimated edges, 
the five performance measures and
the Frobenius norms of differences
of the estimate and true precision matrix
over 50 datasets.
Numbers in the parentheses denote the standard errors.}   \label{tb:prf2}
\medskip
\scalebox{0.85}{
\begin{tabular}{|c|c|c|c|c|c|c|c||c|} \hline
Network & Method & $|\hat{E}|$ & SEN & SPE & FDR & MISR & MCC & $\| \hat{\Omega} - \Omega\|_F$ \\ \hline 
   & \multirow{2}{*}{CLIME} & 643.08 &  91.38 &  99.85 &  29.63 &   0.19 &  80.10 &   8.33 \\
   &   & (2.19) & (0.16) & (0.00) & (0.24) & (0.00) & (0.17) & (0.01) \\ \cline{2-9}
   &	 \multirow{2}{*}{GLASSO} & 810.52 &  92.84 &  99.72 &  43.08 &   0.31 &  72.53 &   7.79 \\
   &   & (7.21) & (0.19) & (0.01) & (0.52) & (0.01) & (0.31) & (0.02) \\ \cline{2-9}
   &	 \multirow{2}{*}{CONCORD} & 682.36 &  93.04 &  99.82 &  32.39 &   0.21 &  79.20 &   5.33 \\
Scale-Free    &   & (4.00) & (0.17) & (0.00) & (0.40) & (0.00) & (0.23) & (0.01) \\ \cline{2-9}
$(|E|=495)$   &	 \multirow{2}{*}{SPMESL-P} & 1022.56 &   94.60 &   99.55 &   54.18 &    0.47 &   65.66 &    5.15 \\
   &   & (3.66) & (0.14) & (0.00) & (0.17) & (0.00) & (0.14) & (0.01) \\ \cline{2-9} 
   &	 \multirow{2}{*}{SPMESL-2} & 457.88 &  87.94 &  99.98 &   4.92 &   0.07 &  91.40 &   6.57 \\
   &   & (1.06) & (0.15) & (0.00) & (0.16) & (0.00) & (0.11) & (0.01) \\ \cline{2-9}
   &	 \multirow{2}{*}{SPMESL-4} & 348.34 &  70.33 & 100.00 &   0.06 &   0.12 &  83.79 &   8.83 \\
   &   & (0.85) & (0.17) & (0.00) & (0.02) & (0.00) & (0.10) & (0.01) \\ \hline 
   &	 \multirow{2}{*}{CLIME} & 644.02 &  84.36 &  99.86 &  27.80 &   0.21 &  77.93 &   8.43 \\
   &   & (2.31) & (0.23) & (0.00) & (0.21) & (0.00) & (0.17) & (0.01) \\ \cline{2-9}
   &	 \multirow{2}{*}{GLASSO} & 881.76 &  87.44 &  99.68 &  45.03 &   0.38 &  69.10 &   7.89  \\
   &   & (10.37) & (0.27) & (0.01) & (0.59) & (0.01) & (0.31) & (0.02) \\ \cline{2-9}
   &	 \multirow{2}{*}{CONCORD} & 731.22 &  89.03 &  99.81 &  32.73 &   0.24 &  77.24 &   5.55 \\
Hub   &   & (5.43) & (0.19) & (0.00) & (0.51) & (0.00) & (0.27) & (0.01) \\ \cline{2-9} 
$(|E|=551)$  &	 \multirow{2}{*}{SPMESL-P} & 1094.14 &   91.32 &   99.52 &   53.99 &    0.51 &   64.62 &    5.37 \\
   &   & (3.33) & (0.14) & (0.00) & (0.14) & (0.00) & (0.12) & (0.01) \\ \cline{2-9} 
   &	 \multirow{2}{*}{SPMESL-2} & 474.36 &  81.22 &  99.98 &   5.64 &   0.10 &  87.49 &   6.78 \\
   &   & (1.55) & (0.22) & (0.00) & (0.15) & (0.00) & (0.13) & (0.01) \\ \cline{2-9}
   &	 \multirow{2}{*}{SPMESL-4} & 319.72 &  57.96 & 100.00 &   0.11 &   0.19 &  76.02 &   8.90 \\
   &   & (1.20) & (0.21) & (0.00) & (0.02) & (0.00) & (0.14) & (0.01)\\ \hline 
\end{tabular}
}
\par\end{centering}
\end{table}

\section{Conclusion}

In this paper,
we proposed an efficient coordinate descent
algorithm with the warm start strategy for sparse precision matrix estimation using the scaled lasso motivated by the empirical observation that
the iterative solution for the diagonal elements of the precision matrix
needs only a few iterations.
In addition, we also develop the parallel coordinate descent algorithm (PCD) for the SPMESL
by representing $p$ Lasso subproblems as the unified minimization problem.
In the PCD algorithm, we use a different convergence criterion
$\| {\bf B}^{(k+1)} - {\bf B}^{(k)}\|_\infty < \delta$ to check
the convergence of the PCD algorithm for the unified minimization problem.
We show that the difference in the iterative solutions of the CD and PCD
caused by the difference of the convergence criteria is also bounded by
the convergence tolerance $\delta$.

Our numerical study shows that the proposed CD algorithm is much
faster than the original algorithm of the SPMESL,
which adopts the LARS algorithm to solve the Lasso subproblems.
Moreover, the PCD algorithm with GPU-parallel computation
becomes more efficient than the CD algorithm when either the number of variables
or the sample size increases. 
For the optimal tuning parameter for the SPMESL, there are three
suggestions without the comparison of the estimation performance.
In the additional simulation,
we numerically investigate the estimation performance of
the SPMESL with the three penalty levels and the other three tuning-search methods.
From the comparison results,
the SPMESL with the uniform bound level and the universal penalty level outperform the three tuning-search methods.
Specifically,
the probabilistic bound level $\lambda_{pb} = \sqrt{2} L_n (k/p)$ provides
the estimate that has the smallest estimation error in terms of the Frobenius norm;
the uniform bound level $\lambda_{ub} = \sqrt{4 n^{-1} \log(p)}$ provides the estimate that
has the smallest FDR less than 1\%; and the universal penalty level $\lambda_{univ} = \sqrt{2 n^{-1} \log(p-1)}$
obtains either the best or second-best in terms of MCC and FDR.
Overall, we recommend the SPMESL with the universal penalty level and apply
the CD algorithm if $p$ is less than or equal to $1000$ and the PCD algorithm
when the target problem has more than 1000 variables and GPU-parallel computation
is available.

\subsection*{Acknowlegements}
Sang Cheol Kim's research was supported by  funding (2019-NI-092-00)
from the Research of Korea Centers for Disease Control and Prevention.
Donghyeon Yu's research was supported by the INHA UNIVERSITY Research Grant
and the National Research Foundation of Korea (NRF-2020R1F1A1A01048127).

\end{document}